\documentclass[11pt]{article}
\usepackage[margin=1in]{geometry}
\usepackage[T1]{fontenc}
\usepackage{lmodern}
\usepackage{amsmath,amssymb,amsthm,mathtools,bm}
\usepackage{microtype,graphicx,booktabs,array,longtable}
\usepackage[numbers,sort&compress]{natbib}
\usepackage{listings}
\newcommand{\doi}[1]{doi:\,\href{https://doi.org/#1}{\nolinkurl{#1}}}
\newcommand{\B}{\mathcal B}
\newcommand{\Ad}{\operatorname{Ad}}
\usepackage[colorlinks=true,allcolors=blue]{hyperref}
\hypersetup{pdftitle={Localized charges, reset noise, and boundary memory in matrix-product-conserving quantum chains},pdfauthor={Ron Rubin},pdfsubject={Matrix-product constraints, statistical localization, and spatially distributed noise}}
\numberwithin{equation}{section}
\newtheorem{theorem}{Theorem}[section]
\newtheorem{proposition}[theorem]{Proposition}
\newtheorem{lemma}[theorem]{Lemma}
\newtheorem{corollary}[theorem]{Corollary}
\theoremstyle{remark}
\newcommand{\A}{\mathcal A}
\newcommand{\E}{\mathbb E}
\newcommand{\Prob}{\mathbb P}
\newcommand{\Tr}{\operatorname{Tr}}
\newcommand{\Var}{\operatorname{Var}}
\newcommand{\diag}{\operatorname{diag}}
\newcommand{\SL}{\mathrm{SL}}
\newcommand{\GL}{\mathrm{GL}}
\newcommand{\id}{\mathrm{id}}
\newcommand{\ip}[2]{\langle #1,#2\rangle_\tau}
\newcommand{\norm}[1]{\left\|#1\right\|}
\title{Localized charges, reset noise, and boundary memory\\in matrix-product-conserving quantum chains}
\author{Ron Rubin\\Rubin Anders Scientific, Inc.}
\date{September 19, 2026}

\begin{document}
\maketitle
\begin{abstract}
We construct boundary observables that retain memory for a quantified time window in quantum chains whose
basis configurations carry a conserved ordered product of matrix labels.
Under strong-irreducibility and proximality hypotheses on the label matrices, projective contraction of random matrix products localizes a conserved charge in the root-mean-square over uniform basis configurations. A finite-time inequality then converts
readout overlap and reset leakage into a correlation bound: for a conserved
reference observable $H$, a readout $F$, accumulated squared leakage $\B$, and the normalized Hilbert--Schmidt inner product,
$\langle F,\Phi(F)\rangle\ge
2\langle H,F\rangle^2/(\|H\|_2^2+\B)-\|F\|_2^2$.
Here $\Phi$ is a sequence of conserving channels and partial resets.
The bound holds for each such circuit, without averaging its gates.
It yields an exponential memory window in the distance from the noisy region,
with quantitative corrections for spatially distributed noise and imperfect
conservation. For a thirteen-state elementary-matrix model, exact polynomial
certificates give root-mean-square localization error $e^{-r/1400}$ between the charge and its restriction to the last $r$ sites, and limiting variance at
least $1/19$ for a rational charge. A related normalized-Gram charge gives
an eight-site certificate retaining more than half the specified readout
correlation through two boundary-reset rounds. Uniform depolarization imposes
an inverse-noise-rate lifetime ceiling. A two-qubit IBM experiment illustrates
the general reset inequality; it does not realize the matrix-product chain.
We provide proofs, exact certificate programs, and an independent recount of
the archived experimental outcomes.
\end{abstract}

\section{Introduction}
How long can a local observable retain information when a distant part of a
quantum system is repeatedly randomized? Conservation laws can obstruct
relaxation, but a useful memory statement needs more than a conserved global
quantity: the information must be preparable and readable in a restricted
region. We study chains whose computational-basis configurations carry an
ordered product of matrix labels. The product is conserved by the interior
dynamics, while a bath acts on a prefix of the chain. A long product of random matrices typically has one singular value far above the others, so it sends almost every input direction close to a single output line. Writing $G_L=AB_r$, the top right singular direction of $G_L$ is therefore nearly determined by the last $r$ labels, and a function of that direction is nearly a function of the suffix alone, up to an error of order $e^{-\kappa r}$. This supplies both a conserved reference observable and a local
preparation and readout with controlled overlap.
A charge here is simply a conserved observable. Localization means that
its root-mean-square difference from a suffix observable decays
exponentially with suffix length, when averaged over uniformly chosen
basis configurations. This is an operator-$L^2$ analogue of statistical localization~\cite{Rakovszky}, with a proved exponential rate. These random configurations, rather than random
quantum gates, are the probability ensemble in the contraction argument.

\paragraph{Reader's guide.}
Each of the $D$ basis states of a site carries an invertible $m\times m$ matrix label. A bounded function of the ordered label product defines a diagonal charge, conserved by gates that mix only configurations with the same product. In the chain application, resets replace a prefix by the maximally mixed state; the untouched suffix is the buffer. Theorem~\ref{thm:local} approximates the charge by a suffix observable with root-mean-square error $C_{\rm loc}e^{-\kappa r}$ under the uniform configuration law. Theorem~\ref{thm:overlapreset} converts its overlap with a readout and the accumulated squared reset leakage into a lower bound on autocorrelation for each prescribed conserving circuit. Together they certify a memory window growing as $e^{2\kappa r}$. The explicit model below supplies numerical constants, finite-chain certificates, and corrections for imperfect conservation and distributed noise.

Such constrained dynamics belongs to the study of Hilbert-space fragmentation:
the Hilbert space separates into many dynamically disconnected sectors;
see the original constructions~\cite{Sala,Khemani}, the commutant-algebra formulation~\cite{MoudgalyaMotrunich}, and the review~\cite{MoudgalyaReview} for examples and broader context.
Han, Chen, and Lake~\cite{Han} and Wang et al.~\cite{Wang} study exponentially
slow thermalization after a boundary coupling breaks fragmentation.
Wang et al. name $\SL(3,\mathbb Z)$ as a nonhyperbolic, nonamenable case beyond their hyperbolic-group argument~\cite[Sec.~V.D]{Wang}. Its nonhyperbolicity follows from a subgroup isomorphic to $\mathbb Z^2$, exhibited in Appendix~\ref{app:global}. Statistically localized integrals of motion
\cite{Rakovszky}, dissipatively enhanced edge-spin correlations
\cite{Vasiloiu}, and fragmentation-breaking local impurities
\cite{Li} provide related mechanisms for persistent local signals.
The construction below gives a mean-square localized charge for matrix-product
constraints and a quantitative response to perturbations at different distances
from the observed edge. The interior gates need not be random.

The second ingredient is a general inequality on operator space. The reference observable is conserved exactly by every factor except the resets, and a reset can
change the overlap with it only by an amount
controlled by the accompanying loss of squared Hilbert--Schmidt norm (a norm budget, not a physical energy). Summing this
norm loss before estimating the final readout produces a budget of
\emph{squared} leakage, which is much smaller than the term-by-term sum $\sum_jp_jb_j$ when each reset disturbs the observable only slightly. This is related in method to sequential-measurement
bounds~\cite{Gao,Oskouei,ODonnellVenkateswaran}, while addressing a different
quantity: the signed correlation at a specified time, with a fixed reference
observable and arbitrary intervening conserving channels. The usual
Mazur--Suzuki bounds concern persistent or time-averaged correlations
\cite{Mazur,Suzuki,DharKunduSaito}; they do not directly give the finite-sequence
inequality proved here. Quantum bottleneck methods~\cite{Bottleneck,Gamarnik}
and quantum Cheeger bounds~\cite[Lemma~20]{Temme} relate small boundary
flow to slow mixing. Applied to a sector-projected state and a conserved
sign charge weakly affected by the bath, a bottleneck argument already
yields a qualitatively exponential memory window~\cite{Bottleneck}.
The present inequality instead controls a specified readout through a
squared-leakage budget for each prescribed, possibly noncommuting channel
sequence; we also derive continuous-time and conservation-defect forms.

There are three levels of conclusions. First, the matrix-product argument
gives exponentially localized charges under standard strong-irreducibility and
proximality hypotheses on the label matrices. Its probability inputs are established results on
products of independent random matrices, from the foundational work of
Furstenberg and Kesten~\cite{FurstenbergKesten,Furstenberg} to quantitative
contraction and limit theorems~\cite{LePage,GuivarchRaugi,GoldsheidMargulis,BougerolLacroix,BenoistQuint};
we use the statements of P\'eneau~\cite{Peneau} and identify them explicitly
before deriving the quantum consequences. Second, for a concrete alphabet of
thirteen $3\times3$ integer matrices, exact rational certificates (polynomial inequalities verified in exact integer or rational arithmetic) supply
numerical constants and a finite eight-site correlation guarantee; the eight sites are thirteen-state sites, not qubits. These
matrices generate $\SL(3,\mathbb Z)$, the group of integer $3\times3$
matrices with determinant one. Unlike the groups covered by the cited
hyperbolic-group argument, this group contains a copy of the integer
plane $\mathbb Z^2$; Appendix~\ref{app:global} gives the specific generators.
The
asymptotic constants are conservative; they do not make multi-thousand-site
readouts (3673 sites for the fixed-image charge, 5211 for the normalized-Gram charge) an experimentally efficient proposal. Third, a two-qubit experiment
illustrates an exactly soluble instance of the general reset/readout inequality.
It establishes neither the spatial localization nor the long-chain lifetime.

To our knowledge, the contribution is the matrix-product charge construction,
its application to the specified nonhyperbolic family, and its quantitative
spatial perturbation and finite-readout bounds. Qualitative non-expansion
of group heat kernels, meaning that under the law of a product of $L$ random letters there are sets of substantial mass whose boundary in the Cayley graph has small mass, is known~\cite{FraczykVanLimbeek}; Appendix~\ref{app:global}
gives, for the family considered here, an exponentially small vertex-boundary mass and a statement that typical basis states stay far from equilibrium for exponentially many steps of the averaged dynamics. Retaining a classical observable
is distinct from recovering an arbitrary quantum state, a distinction illustrated
in Section~\ref{sec:discussion} by a dephasing example, with pointers to the recovery bounds of
Tyson~\cite{Tyson} and Barnum and Knill~\cite{BarnumKnill}.

The argument proceeds from localization (Section~\ref{sec:localization}) to a channel inequality (Section~\ref{c-sec:setting}), then to
noise and lifetime bounds. Section~\ref{sec:explicit-summary} summarizes the
explicit constants, with their polynomial certificates in the appendices.
Section~\ref{sec:gram} derives the finite-chain certificate. The final
experimental sections give the circuit, estimator, raw counts, and scope of
the hardware comparison. Appendix~\ref{app:reproducibility} specifies how to
reproduce the calculations without access to a quantum device.

\section{Matrix labels, quantum dynamics, and the memory signal}\label{sec:model}
Let $\A\subset\GL_m(\mathbb R)$ be a finite alphabet of size $D$, $m\ge2$.
Here $\GL_m(\mathbb R)$ denotes the invertible real $m\times m$ matrices;
each matrix in $\A$ labels one of a site's $D$ basis states.
We borrow the words alphabet, letter and word from formal languages. A word $w=(a_1,\ldots,a_L)$, a string of $L$ letters, labels the basis state $|a_1\rangle\otimes\cdots\otimes|a_L\rangle$ of
$\mathcal H_L=(\mathbb C^D)^{\otimes L}$ and carries
\begin{equation}
 G_L(w)=a_1\cdots a_L,\qquad B_r(w)=a_{L-r+1}\cdots a_L.
\end{equation}
The labels are classical data attached to basis states, not physical gates. A product-conserving unitary is one that is block
diagonal in the equal-$G_L$ basis fibres defined below. All contiguous equal-subword-product
gates, which map each basis state into the span of words differing only on a block of adjacent sites with the same block product, have this property. The reference measure is the uniform word law
$\pi_L$, which gives every word probability $D^{-L}$ and corresponds to the maximally mixed state $I/D^L$; the operator inner product is
\begin{equation}
 \tau_L(X)=D^{-L}\Tr X,\qquad
 \ip XY=\tau_L(X^\dagger Y),\qquad \norm X_2^2=\ip XX.
\end{equation}
All operator norms without a subscript $2$ or $F$ are spectral norms.
The spectral norm $\|X\|=\|X\|_\infty$ is the largest singular value.
For label matrices, $\|G\|_F^2=\Tr(G^TG)$ is the unnormalized
Frobenius norm; vector norms are Euclidean. This differs from the
normalized Hilbert--Schmidt norm on physical operators. The unnormalized trace norm of a physical operator is written $\|\cdot\|_{\mathrm{tr}}$.
We call a quantum channel bistochastic when it is completely positive,
trace preserving and unital, meaning $T(I)=I$. Such channels contract
$\norm\cdot_2$, as proved in Section~\ref{c-sec:setting}.

For example, if the alphabet contains $I,a,a^{-1}$, a two-site gate may
mix $|I,I\rangle$, $|a,a^{-1}\rangle$ and $|a^{-1},a\rangle$: all three have product $I$. The alphabet~\eqref{eq:alphabet} contains such inverse pairs.
Any unitary rotation on their span, acting as the identity on its orthogonal
complement, preserves the product. Thus the constraint permits coherent
quantum dynamics even though the conserved labels are classical matrices.
More generally, each function $u(G_L)$ defines a diagonal conserved
observable, an element of the commutant of the product-conserving gate
algebra in the sense of Ref.~\cite{MoudgalyaMotrunich}. The problem is to choose $u$ so that this global observable
can be well approximated near the right edge.
More explicitly, an equal-product fibre is
$\mathcal V_g=\operatorname{span}\{|w\rangle:G_L(w)=g\}$.
The observable $\operatorname{diag}(u(G_L))$ equals $u(g)I$ on
$\mathcal V_g$, so it commutes with every unitary preserving each
such subspace. Replacing a contiguous subword by another with the same
product leaves the complete ordered product unchanged by associativity.

For diagonal $H=\operatorname{diag}(h(w))$, the Hilbert--Schmidt norm
is an ordinary root-mean-square over words:
\begin{equation}
 \|H\|_2^2=D^{-L}\sum_{w\in\A^L}|h(w)|^2.
\end{equation}
Here $\E$ is expectation over independent uniform letters, corresponding
to the physical reference state $I/D^L$. Subtracting $\E h$ removes the
equilibrium signal and makes the squared norm equal to the variance.
A suffix readout
depends only on the last $\ell$ symbols and is tensored with the identity
on the other sites, $F=I^{\otimes(L-\ell)}\otimes f$; it therefore commutes with every operator, and is fixed by every unital channel, supported on the first $L-\ell$ sites. Site $1$ is the bath end and site $L$ the readout end. Throughout, $m$ is the dimension of a label matrix,
$D$ is the number of physical states per site, $L$ is the chain length,
$\ell$ is the readout width, and $r$ is the number of rightmost sites
untouched by a given perturbation. The symbol $N$ in the general channel
inequality denotes the full Hilbert-space dimension ($D^L$ for the chain).

\section{From random matrix products to a localized charge}\label{sec:localization}
For a nonzero vector $x$, the notation $[x]$ denotes its one-dimensional
subspace, or line, so that
$x$ and any nonzero scalar multiple specify the same point. The distance
between lines is
\begin{equation}
 d([x],[y])=\sqrt{1-\frac{(x^Ty)^2}{\|x\|^2\|y\|^2}}.
\end{equation}
This is the sine of the angle between the two lines, and it is a metric on the set of lines. Let $s_1(G)\ge s_2(G)\ge\cdots\ge s_m(G)>0$ denote the singular values of $G$. A matrix with $s_1$ much larger than $s_2$ maps most
directions close to one line. The hypotheses below guarantee such
contraction for typical long products, while excluding a fixed finite
collection of invariant subspaces.

A semigroup is the set of all finite products of its generators.
It is \emph{strongly irreducible} if it preserves no finite union of
nonzero proper linear subspaces, and \emph{proximal} if it contains a
matrix with a simple eigenvalue strictly largest in modulus. These are
properties of the label matrices, not of physical quantum gates.
The underlying theory of Lyapunov exponents and stationary projective
measures originates in Refs.~\cite{FurstenbergKesten,Furstenberg}; see
also Ref.~\cite{BougerolLacroix}. The specific probability estimates used
here are stated next.
Assume that both $\A$ and its transpose alphabet $\A^T$ generate
strongly irreducible proximal semigroups. We use the following standard
random-product consequences~\cite[Thms.~1.17 and 1.20]{Peneau}.
Let $s_1(G)\ge s_2(G)\ge\cdots>0$ be the singular values and
$G_n$ the product of $n$ independent uniform letters.
Writing $\sigma(G)=s_2(G)/s_1(G)$, there are positive constants
$a_0,b_0,C_1,\alpha,\beta,C_2$ such that
\begin{equation}\label{eq:gap}
 \Prob\{\sigma(G_n)>e^{-a_0n}\}\le C_1e^{-b_0n}.
\end{equation}
For an infinite iid transpose sequence $M_1,M_2,\ldots$, a common random
projective line $\Lambda$ satisfies
\begin{equation}\label{eq:common}
 \Prob\{d([M_1\cdots M_nx],\Lambda)>e^{-\alpha n}\}
 \le C_2e^{-\beta n},\qquad x\ne0,
\end{equation}
uniformly in deterministic $x$. The law
of the input letters is uniform on $\A^T$. Crucially, $\Lambda$ depends
on the sequence, not on the chosen input $x$. The law
of $\Lambda$ is the stationary measure of the transpose walk, the Markov chain on lines $[x]\mapsto[Mx]$ with $M$ uniform on $\A^T$. Independence
allows conditioning on random inputs in~\eqref{eq:common}.
Stationarity means that multiplying an independent line of this law by
one random transpose letter leaves its distribution unchanged. The two
estimates above are imported probability theorems; the localization and
channel arguments below are derived from them. Equation~\eqref{eq:gap}
uses Theorem~1.17 of Ref.~\cite{Peneau}, and~\eqref{eq:common} uses its
Theorem~1.20. Invertibility ensures positive eventual rank, as required
there; the quoted results do not require a moment assumption.

Let $v$ be a unit eigenvector of $G^TG$ for its largest eigenvalue $s_1(G)^2$, a top right singular vector, and let $Z(G)=vv^T$ be a Borel (measurable) choice of this projector; then $Z(G)_{jj}=v_j^2$ is the squared $j$-th coordinate of the dominant right singular direction. At a
degeneracy, select the first nonzero coordinate projection into the top
eigenspace of $G^TG$ and normalize it. This fixes $Z$ as a function of $G$
alone. Degeneracies belong to the exceptional event in~\eqref{eq:gap}.
For a coordinate $j$, define
\begin{equation}\label{eq:charges}
 \begin{split}
 H_L&=\diag\big(Z(G_L)_{jj}-\E Z(G_L)_{jj}\big),\\
 F_r&=\diag\big(Z(B_r)_{jj}-\E Z(B_r)_{jj}\big).
 \end{split}
\end{equation}
$F_r$ is extended by the identity outside the last $r$ sites; the coordinate $j$ is fixed throughout, and the theorem below says how to choose it. Write
$v_L=\norm{H_L}_2^2$ and $v_r=\norm{F_r}_2^2$. Both operators are centered,
have spectral norm at most one, and have $2$-norm at most $1/2$.
The latter follows because they center random variables in $[0,1]$, whose variance is at most $1/4$.

\subsection{Mean-square localization and its proof}
Define a spherical small-ball constant by
\begin{equation}
 K_2=1,\qquad
 K_m=\frac{2\Gamma(m/2)}{\sqrt\pi\,\Gamma((m-1)/2)}\quad(m\ge3).
\end{equation}
These constants satisfy $K_m\ge1$, so the small-ball bound below also holds trivially for $s>1$.
For a uniform unit vector $Y$ and any unit vector $u$,
$\Prob(|u^TY|\le s)\le K_ms$ for $0\le s\le1$.
For $m=2$ this is $2\arcsin(s)/\pi\le s$; for $m\ge3$ the
coordinate density is
$f_m(t)=(K_m/2)(1-t^2)^{(m-3)/2}$ on $[-1,1]$, whose maximum is at zero.
Here $\Gamma$ is Euler's gamma function, and $K_m\ge1$ because a probability density on an interval of length two has maximum at least $1/2$.

\begin{theorem}[Mean-square localization]\label{thm:local}
Set
\begin{equation}\label{eq:constants}
 \kappa=\tfrac12\min\{a_0,b_0,2\alpha,\beta\},\qquad
 C_{\rm loc}=4\sqrt{K_m(1+C_1)}+2\sqrt{2+C_2}.
\end{equation}
For all $L\ge r\ge1$,
\begin{equation}\label{eq:local}
 d_{L,r}:=\norm{H_L-F_r}_2\le
 \delta_r:=\min\{1,C_{\rm loc}e^{-\kappa r}\}.
\end{equation}
For at least one coordinate $j$, the limiting variance
$v_*:=\lim_{L\to\infty}v_L$ exists and is positive. Moreover,
\begin{equation}\label{eq:varfinite}
 |\sqrt{v_r}-\sqrt{v_*}|\le C_{\rm loc}e^{-\kappa r}.
\end{equation}
\end{theorem}
\begin{proof}
The proof uses an auxiliary random direction $Y$ to compare a singular
projector with the image of a vector. The whole product and its suffix
then send their inputs close to the same random line. Adding these
three approximation errors gives the stated localization constant;
$Y$ is a proof device and is not a physical ancilla.
Write $Q=s_1uv^T+E$, with $Ev=0$ and $\norm E=s_2$. Then
\begin{equation}
 Q^TY=s_1(u^TY)v+E^TY,\qquad E^TY\perp v,
 \qquad \|Q^TY\|\ge s_1|u^TY|.
\end{equation}
Consequently
\begin{equation}
 d([Q^TY],[v])\le \min\{1,\sigma(Q)/|u^TY|\}.
\end{equation}
For a random variable $W\in[0,1]$,
\begin{equation}
 \E W=\int_0^1\Prob(W>t)\,dt.
\end{equation}
Apply this identity with $W=\min\{1,s^2/|u^TY|^2\}$ and use
$\Prob(W>t)\le K_ms/\sqrt t$. Integration gives, for $0\le s\le1$,
\begin{equation}
 \E_Y\min\{1,s^2/|u^TY|^2\}\le 2K_ms.
\end{equation}
Writing $\mathcal Z(x)=xx^T/\norm x^2$ and using
$\norm{\mathcal Z(x)-\mathcal Z(y)}_F^2=2d([x],[y])^2$, we obtain
first the conditional estimate
$\E_Y\|Z(Q)-\mathcal Z(Q^TY)\|_F^2\le4K_m\sigma(Q)$.
Splitting the matrix expectation into the good and exceptional events
of~\eqref{eq:gap} gives
$\E\sigma(G_n)\le e^{-a_0n}+C_1e^{-b_0n}$, and hence
\begin{equation}\label{eq:sphere}
 \E\norm{Z(G_n)-\mathcal Z(G_n^TY)}_F^2
 \le4K_m(1+C_1)e^{-\min(a_0,b_0)n}.
\end{equation}

For $G_L=AB_r$, where $A=a_1\cdots a_{L-r}$ (or $I$ if $L=r$),
take $Y$ independent and uniform, and put
$X=A^TY/\norm{A^TY}$. The pair $(X,Y)$ is independent of the suffix.
In reverse transpose order, $M_j=a_{L-j+1}^T$, so
$M_1\cdots M_r=B_r^T$.
Extend this sequence by iid transpose letters independent of $(A,Y)$
and of the suffix, so that the infinite-sequence theorem applies.
Apply~\eqref{eq:common} to both inputs with the same $\Lambda$.
The probability that their projective separation exceeds
$2e^{-\alpha r}$ is at most $2C_2e^{-\beta r}$. Consequently
\begin{equation}
 \E\norm{\mathcal Z(G_L^TY)-\mathcal Z(B_r^TY)}_F^2
 \le8e^{-2\alpha r}+4C_2e^{-\beta r}.
\end{equation}
For a random matrix write $\|X\|_{L^2(F)}=(\E\|X\|_F^2)^{1/2}$; this is a norm, so the triangle inequality below is Minkowski's inequality.
The triangle inequality gives
\begin{align}
 \|Z(G_L)-Z(B_r)\|_{L^2(F)}
 &\le\|Z(G_L)-\mathcal Z(G_L^TY)\|_{L^2(F)}\nonumber\\
 &\quad+\|\mathcal Z(G_L^TY)-\mathcal Z(B_r^TY)\|_{L^2(F)}\nonumber\\
 &\quad+\|\mathcal Z(B_r^TY)-Z(B_r)\|_{L^2(F)}.
\end{align}
The first and third terms are each at most
$2\sqrt{K_m(1+C_1)}e^{-\min(a_0,b_0)r/2}$; the middle term is at most
$2\sqrt{2+C_2}e^{-\min(\alpha,\beta/2)r}$.
Their sum proves
$\norm{Z(G_L)-Z(B_r)}_{L^2(F)}\le C_{\rm loc}e^{-\kappa r}$.
Taking a diagonal coordinate and centering are contractions in real $L^2$,
which proves~\eqref{eq:local}; the additional bound by one follows from
$\norm{H_L}_2,\norm{F_r}_2\le1/2$.

Equations~\eqref{eq:common} and~\eqref{eq:sphere} identify the limiting
projector law $\eta$ as the law of the projector onto $\Lambda$.
Bounded coordinates imply convergence of their first and second moments,
so $v_L$ converges for each $j$.
The law is nonatomic, meaning no individual line has positive probability:
otherwise its atoms of largest mass
form a nonempty finite set, and stationarity forces every transpose letter
to permute this set, contradicting strong irreducibility. If all diagonal
variances vanished, all squared coordinates of its rank-one projectors
would be fixed. Only finitely many sign patterns would remain, contradicting
nonatomicity. Thus some $j$ has $v_*>0$. Finally,
$|\sqrt{v_L}-\sqrt{v_r}|\le d_{L,r}$; let $L\to\infty$ to
obtain~\eqref{eq:varfinite}.
\end{proof}

The construction gives mean-square localization in the reference ensemble.
It does not give uniform operator-norm localization over all configurations.
For an identity product, $\sigma(G_L)=1$, so the contraction argument gives no uniform estimate. More strongly, for the alphabets below one may choose a long anisotropic suffix and an inverse prefix so that $G_L=I$ while $Z(B_r)$ stays an order-one distance from the fixed choice $Z(I)$. Thus the exponential approximation cannot hold uniformly in operator norm over all words.
Appendices~\ref{sec:explicit} and~\ref{sec:sharper} give numerical constants
for a different charge in the thirteen-state model.

\section{Channels, readouts, and reset leakage}
\label{c-sec:setting}
Let the Hilbert space have dimension $N$, and use normalized trace and
the Hilbert--Schmidt inner product
\begin{equation}
 \tau(X)=N^{-1}\Tr X,\qquad
 \ip XY=\tau(X^\dagger Y),\qquad
 \norm X_2^2=\ip XX.
\end{equation}
The norm $\norm\cdot_\infty$ denotes the operator norm. A
\emph{bistochastic} channel is completely positive, trace preserving,
and unital. Such a channel $T$ contracts $\norm\cdot_2$: the
operator Schwarz inequality gives
$T(X)^\dagger T(X)\le T(X^\dagger X)$, and taking the trace proves
$\norm{T(X)}_2\le\norm X_2$ \cite{Kadison,Choi,Watrous}.

Let $H$ be a nonzero Hermitian reference observable. It is not necessarily
the physical Hamiltonian, which we denote by $K$ in Section~\ref{c-sec:extensions}. A conserving
channel satisfies $T(H)=H$. For unitary conjugation
$T=\Ad_U$, $\Ad_U(X)=UXU^\dagger$, this is $[U,H]=0$.
Let $P_j$ be bistochastic channels that are orthogonal projections on
operator space:
\begin{equation}
 P_j^2=P_j=P_j^*,\qquad
 Q_j=(1-p_j)\id+p_jP_j,\quad 0\le p_j\le1.
 \label{c-eq:partial}
\end{equation}
Here $*$ is the adjoint for $\ip\cdot\cdot$. Complete randomization of
a subsystem $A$ is the principal example:
\begin{equation}
 P_A(X)=\frac{I_A}{d_A}\otimes\Tr_A X.
 \label{c-eq:reset}
\end{equation}
Throughout, ``reset'' means this unital randomization, not preparation
of a pure state such as $|0\rangle$. Dephasing channels are another
class of orthogonal projection channels covered by the argument.

The full channel $\Phi$ is any finite composition of the $Q_j$ and
conserving bistochastic factors, in their prescribed order. Neither
the projections nor the other factors need commute. For a Hermitian
readout $F$, define
\begin{equation}
 \begin{split}
 v&=\norm H_2^2>0,\qquad w=\norm F_2^2,\qquad a=\ip HF,\\
 b_j^2&=\norm{(\id-P_j)H}_2^2,\qquad
 \B=\sum_j\frac{p_j}{2-p_j}b_j^2,\\
 C_F&=\ip F{\Phi(F)}.
 \end{split}
 \label{c-eq:parameters}
\end{equation}
These quantities are real. The index $j$ runs over the reset occurrences along the circuit, so a repeated projection is counted each time it is applied. All channels are written in the
Schr\"odinger picture, extended linearly to operators; equivalently,
$C_F=\tau[F\Phi^*(F)]$.
Centering $F$ and $H$ is unnecessary for the inequality. For a memory
interpretation we use centered observables, removing the constant
equilibrium contribution.

For example, if $\tau(F)=0$ and $|\epsilon|\norm F_\infty\le1$,
the state $\rho_\epsilon=(I+\epsilon F)/N$ satisfies
\begin{equation}
 \Tr[F\Phi(\rho_\epsilon)]=\epsilon C_F.
 \label{c-eq:response}
\end{equation}
Thus $C_F$ is the retained linear signal from a preparation aligned with
$F$. For a binary readout it is also an initial--final sign correlation,
as derived in Section~\ref{c-sec:experiment}.

\subsection{The readout-overlap inequality}
\label{c-sec:theorem}
\begin{theorem}[Finite-sequence correlation bound]
\label{thm:overlapreset}
Under the hypotheses of Section~\ref{c-sec:setting},
\begin{equation}
 \boxed{\displaystyle C_F\ge\frac{2a^2}{v+\B}-w.}
 \label{c-eq:main}
\end{equation}
The readout $F$ need not commute with $H$. No locality, independence
of the factors, or average over their choices is assumed.
\end{theorem}

We separate the proof into a dissipative estimate and a geometric step.
This also identifies where exact conservation and the projection
structure enter.

\begin{lemma}[Overlap loss is paid for by norm loss]
\label{c-lem:loss}
Writing $\psi=\Phi(F)$, one has
\begin{equation}
 \left|a-\ip H\psi\right|^2
 \le\B\bigl(w-\norm\psi_2^2\bigr).
 \label{c-eq:loss}
\end{equation}
\end{lemma}
\begin{proof}
First, a contraction $T$ fixing $H$ also satisfies $T^*H=H$.
Indeed,
$\ip{T^*H}H=\ip H{TH}=v$ and
$\norm{T^*H}_2\le\sqrt v$; equality in Cauchy--Schwarz implies the
claim. Hence a conserving factor does not change the overlap with $H$.

At a reset step, let $X$ be the incoming operator and put
$R_j=\id-P_j$. Orthogonality gives
\begin{align}
 Q_jX&=P_jX+(1-p_j)R_jX,\\
 d_j:=\norm X_2^2-\norm{Q_jX}_2^2
 &=p_j(2-p_j)\norm{R_jX}_2^2,\label{c-eq:normloss}\\
 \ip H{X-Q_jX}&=p_j\ip{R_jH}{R_jX}.
 \label{c-eq:overlapstep}
\end{align}
For $p_j>0$, the last line is bounded in absolute value by
$\sqrt{p_j/(2-p_j)}\,b_j\sqrt{d_j}$; the same statement is
zero on both sides when $p_j=0$. Telescoping the overlap across the
actual trajectory and applying Cauchy--Schwarz to the sum gives
\begin{equation}
 \left|a-\ip H\psi\right|
 \le\sum_j\sqrt{\frac{p_j}{2-p_j}}\,b_j\sqrt{d_j}
 \le\sqrt{\B\sum_jd_j}.
\end{equation}
Every intervening factor is a contraction, so the reset losses satisfy
$\sum_jd_j\le w-\norm\psi_2^2$. This proves~\eqref{c-eq:loss}.
\end{proof}

\begin{proof}[Proof of Theorem~\ref{thm:overlapreset}]
Suppose first that $\B>0$, and set $h=H/\sqrt v$. Decompose
\begin{equation}
 F=\frac a{\sqrt v}h+F_\perp,\qquad
 \psi=z h+\psi_\perp,\qquad
 \ip h{F_\perp}=\ip h{\psi_\perp}=0.
\end{equation}
Since all operators involved are Hermitian, $z$ is real.
Equation~\eqref{c-eq:loss} reads
\begin{equation}
 (a-\sqrt v\,z)^2\le\B
       (w-z^2-\norm{\psi_\perp}_2^2).
\end{equation}
Completing the square, and writing $R=w-a^2/(v+\B)\ge0$, yields
\begin{equation}
 \frac{v+\B}{\B}
 \left(z-\frac{a\sqrt v}{v+\B}\right)^2
 +\norm{\psi_\perp}_2^2\le R.
 \label{c-eq:ellipsoid}
\end{equation}
Put $u=\sqrt{(v+\B)/\B}\,[z-a\sqrt v/(v+\B)]$. Then
\begin{equation}
 C_F=\frac{a^2}{v+\B}
   +\frac a{\sqrt v}\sqrt{\frac{\B}{v+\B}}\,u
   +\ip{F_\perp}{\psi_\perp}.
\end{equation}
The squared norm of the coefficients of the last two terms is
\begin{equation}
 \frac{a^2\B}{v(v+\B)}+\norm{F_\perp}_2^2
 =\frac{a^2\B}{v(v+\B)}+w-\frac{a^2}v=R.
\end{equation}
Cauchy--Schwarz and~\eqref{c-eq:ellipsoid} therefore bound those terms
below by $-R$, proving $C_F\ge2a^2/(v+\B)-w$.

If $\B=0$, Lemma~\ref{c-lem:loss} fixes $\ip H\psi=a$.
Contractivity gives
$\norm{\psi_\perp}_2^2\le w-a^2/v=\norm{F_\perp}_2^2$,
and hence
$C_F=a^2/v+\ip{F_\perp}{\psi_\perp}\ge2a^2/v-w$.
The case $F=0$, where both sides vanish, is included.
\end{proof}

The coefficient $p/(2-p)$ is the ratio between the squared overlap
change and the available norm loss at a partial reset. Replacing
each step by a triangle-inequality estimate would instead involve
$\sum_jp_jb_j$. The squared leakage is useful precisely when the
reference observable is only weakly affected by each reset.

\subsection{Normalized form and retention times}
For $w>0$, introduce the overlap fraction and relative leakage
\begin{equation}\label{c-eq:qx}
 q=\frac{a^2}{vw}\in[0,1],\qquad x=\frac{\B}{v}.
\end{equation}
The bound becomes
\begin{equation}
 \frac{C_F}{w}\ge\frac{2q}{1+x}-1.
 \label{c-eq:normalized}
\end{equation}
Thus positive guaranteed correlation requires both sufficient overlap
and sufficiently small leakage. Even at $x=0$, a readout with
$q\le1/2$ receives no positive guarantee from this estimate: its
orthogonal component can contribute negatively at a particular time.
This explains why~\eqref{c-eq:normalized} is not the usual time-averaged
Mazur bound $\overline{C_F}/w\ge q$.

For $F=H$, equation~\eqref{c-eq:main} gives
\begin{equation}
 C_H\ge v\frac{v-\B}{v+\B}\ge v-2\B.
 \label{c-eq:aligned}
\end{equation}
Call a round one pass of the conserving dynamics together with its resets. If each round contributes at most $\beta_{\rm round}>0$ to $\B$, a target
retention fraction $\theta\in(0,1)$ is guaranteed after $n$ rounds whenever
\begin{equation}
 n\beta_{\rm round}\le v\left(\frac{2q}{1+\theta}-1\right),
 \qquad q>\frac{1+\theta}{2}.
 \label{c-eq:lifetime}
\end{equation}
This is a sufficient finite-time condition, not a statement about
the sharp relaxation time.

For a spatial application, suppose a local approximation $F_\ell$ to
$H$ satisfies $\norm{H-F_\ell}_2\le\delta_\ell$, and a reset
outside its support fixes $F_\ell$. Then
\begin{equation}
 \norm{(\id-P_j)H}_2
 =\norm{(\id-P_j)(H-F_\ell)}_2\le\delta_\ell.
\end{equation}
If $\delta_\ell\le C_{\rm loc}e^{-\kappa\ell}$ and a round contains $J$
complete resets of this kind, then
$\beta_{\rm round}\le J C_{\rm loc}^2e^{-2\kappa\ell}$.
For a fixed positive overlap margin in~\eqref{c-eq:lifetime}, this
certifies a number of rounds proportional to $e^{2\kappa\ell}$.
This consequence is conditional on establishing localization and
overlap in a particular model; Lemma~\ref{lem:spatial} and Corollary~\ref{cor:prefix} do this for the chain. The two-qubit experiment supplies
neither a spatial exponent nor a many-body lifetime.

\subsection{Replacing the conserved charge by a local readout}
For the chain, use $H=H_L$, write $v=v_L$, and define
$C_\ell=\ip{F_\ell}{\Phi(F_\ell)}$. Set
\begin{equation}
 E_\ell=d_{L,\ell}(\sqrt{v_L}+\sqrt{v_\ell})\le d_{L,\ell},
\end{equation}
where the inequality uses $\sqrt{v_L},\sqrt{v_\ell}\le1/2$.
For a reset with buffer $r_j$, meaning that it acts only on the first
$L-r_j$ sites, $P_jF_{r_j}=F_{r_j}$ and hence
$b_j=\|(\id-P_j)H_L\|_2\le\delta_{r_j}$.
The error $E_\ell$ controls the replacement of both endpoints of the correlation:
for any bistochastic $\Phi$, adding and subtracting $\langle H_L,\Phi(F_\ell)\rangle_\tau$ and using contractivity,
\begin{align}
 |\langle F_\ell,\Phi(F_\ell)\rangle_\tau
    -\langle H_L,\Phi(H_L)\rangle_\tau|
 &\le \|F_\ell-H_L\|_2\|\Phi(F_\ell)\|_2
       +\|H_L\|_2\|\Phi(F_\ell-H_L)\|_2\nonumber\\
 &\le E_\ell.
\end{align}
\begin{corollary}[Reset bound for the local readout]\label{thm:reset}
For the channel sequence of Theorem~\ref{thm:overlapreset} conserving
$H_L$ between resets, with $v_L>0$,
\begin{equation}\label{eq:reset}
 C_\ell\ge v_L\frac{v_L-\B}{v_L+\B}-E_\ell
          \ge v_L-2\B-E_\ell.
\end{equation}
\end{corollary}
\begin{proof}
Apply~\eqref{c-eq:aligned} and the endpoint estimate above.
\end{proof}

\section{Distributed perturbations and finite-size certification}
For $r=0$ set $F_0=0$ and $\delta_0=\sqrt{v_L}\le1/2$.
An operator or bistochastic channel has buffer $r$ if its support lies in the first
$L-r$ sites, that is, $V=V'\otimes I$ or $\mathcal E=\mathcal E'\otimes\id$ with the second factor on the last $r$ sites. It then commutes with, or fixes, respectively, the suffix
operator $F_r=I\otimes f_r$: indeed $\mathcal E(F_r)=\mathcal E'(I)\otimes f_r=F_r$ because $\mathcal E'$ is unital.

\begin{lemma}[Spatial leakage]\label{lem:spatial}
Let $\mathcal E$ be a bistochastic channel with buffer $r$. Then
\begin{equation}\label{eq:channeldefect}
 \norm{(\mathcal E-\id)H_L}_2\le2\delta_r.
\end{equation}
If $P$ is a bistochastic orthogonal projection channel with buffer $r$,
$\norm{(\id-P)H_L}_2\le\delta_r$. A Hermitian perturbation $V$ with
buffer $r$ obeys
\begin{equation}\label{eq:commdefect}
 \norm{[V,H_L]}_2\le2\norm V\delta_r.
\end{equation}
\end{lemma}
\begin{proof}
Replace $H_L$ by $H_L-F_r$ in each expression, which is allowed because $\mathcal E$, $P$ and $V$ all fix or commute with $F_r$. Then $\|(\mathcal E-\id)(H_L-F_r)\|_2\le\|\mathcal E(H_L-F_r)\|_2+\|H_L-F_r\|_2\le2d_{L,r}$ by channel contractivity and the triangle inequality; $\id-P$ is again an orthogonal projection, hence a contraction; and $\|VX\|_2,\|XV\|_2\le\|V\|\,\|X\|_2$ with $\|V\|$ the spectral norm (the ideal property of the Hilbert--Schmidt norm) gives the commutator bound.
\end{proof}

Let $\Phi_t$ be the bistochastic channel after $t$ rounds (or after time $t$ in the continuous-time setting later on), acting on operators by linearity as in Section~\ref{c-sec:setting}, and let $\Phi_t^*$ be its adjoint for $\ip\cdot\cdot$, the Heisenberg picture. Define the
infinite-temperature autocorrelation, the correlation in the reference state $I/D^L$,
\begin{equation}
 C_\ell(t)=\tau_L\big(F_\ell\Phi_t(F_\ell)\big)
           =\tau_L\big(F_\ell\Phi_t^*(F_\ell)\big).
\end{equation}
This is real because $\Phi_t$ preserves Hermiticity and $\tau_L(AB)$ is real for Hermitian $A,B$; it need not be positive without further information.

\begin{theorem}[Arbitrary channels and a finite-block certificate]\label{thm:linear}
Let $\Phi=\mathcal E_n\cdots\mathcal E_1$ be a product of bistochastic
channels and put $e_j=\norm{(\mathcal E_j-\id)H_L}_2$. Then
\begin{equation}\label{eq:finitecert}
 C_\ell\ge v_\ell-\sqrt{v_\ell}
 \left(2d_{L,\ell}+\sum_{j=1}^n e_j\right).
\end{equation}
Conserving factors have $e_j=0$. A factor
$\mathcal E_j=(1-p_j)\id+p_j\mathcal V_j$, $0\le p_j\le1$,
where $\mathcal V_j$ is bistochastic with buffer $r_j$, has
$e_j\le2p_j\delta_{r_j}$. No independence or reversibility is required.
\end{theorem}
\begin{proof}
Expand $\Phi H_L-H_L=\sum_{j=1}^n\mathcal E_n\cdots\mathcal E_{j+1}(\mathcal E_j-\id)H_L$, each factor's defect applied to $H_L$ itself and propagated
through the later factors. Contractivity of $\mathcal E_n\cdots\mathcal E_{j+1}$ gives a norm bound by $\sum e_j$.
Also $\Phi F_\ell-F_\ell=\Phi(F_\ell-H_L)+(\Phi H_L-H_L)+(H_L-F_\ell)$, so $\norm{\Phi F_\ell-F_\ell}_2\le2d_{L,\ell}+\sum e_j$.
Take the inner product with $F_\ell$ and use Cauchy--Schwarz with $\|F_\ell\|_2=\sqrt{v_\ell}$.
\end{proof}

For example, group the factors into $t$ consecutive rounds and suppose the factors of each round have $\sum_je_j\le\varepsilon_L$, the subscript recording that the defects depend on $L$ through the buffers. Then $\sum_je_j\le t\varepsilon_L$, and
$C_\ell(t)\ge\theta v_\ell$, $0<\theta<1$, throughout
\begin{equation}\label{eq:finitewindow}
 t\le\frac{(1-\theta)\sqrt{v_\ell}-2d_{L,\ell}}{\varepsilon_L},
 \qquad (1-\theta)\sqrt{v_\ell}>2d_{L,\ell}.
\end{equation}
The finite-block variance reduces to a sum over $\ell$-letter words because $F_\ell$ depends only on the last $\ell$ letters; with $G_\ell(w)=a_1\cdots a_\ell$,
\begin{equation}
 v_\ell=D^{-\ell}\sum_{w\in\A^\ell}
 \left(Z(G_\ell(w))_{jj}-D^{-\ell}\sum_{u\in\A^\ell}
 Z(G_\ell(u))_{jj}\right)^2.
\end{equation}
Unlike $v_*$ of Theorem~\ref{thm:local}, $v_\ell$ is a finite sum of $D^\ell$ terms, so~\eqref{eq:finitecert} is a certificate at fixed finite $L$ with no limit taken.
Explicit numerical bounds on $d_{L,\ell}$ and $\varepsilon_L$ turn this
inequality into a numerical certificate. If
$\varepsilon_L\le C e^{-\lambda L}$, it guarantees a number of integer
rounds proportional to $e^{\lambda L}$, provided the numerator has a
positive lower bound uniform in $L$.
When $\varepsilon_L=0$ and the numerator is positive, \eqref{eq:finitewindow} imposes no restriction on $t$.

\section{Continuous dynamics and conservation defects}
Work in the Schr\"odinger picture on the operator Hilbert space, the space of operators with the inner product $\ip\cdot\cdot$; $*$ denotes the adjoint for it, and $\mathcal D_s\ge0$ means $\ip X{\mathcal D_sX}\ge0$ for all $X$. Suppose
\begin{equation}\label{eq:lindblad}
 \dot X=(\mathcal A_s-\mathcal D_s)X,
 \quad\mathcal A_s^*=-\mathcal A_s,\quad
 \mathcal D_s^*=\mathcal D_s\ge0,
\end{equation}
with a bistochastic propagator. An important instance is
\begin{equation}
 \mathcal A_s X=-i[K_s,X],\qquad
 \mathcal D_s=\sum_j\varepsilon_j(s)(\id-P_j).
\end{equation}
Here $\varepsilon_j(s)\ge0$ are reset rates, $K_s$ is Hermitian (so that $\mathcal A_s$ is skew-adjoint), and each $P_j$ is a projection channel as in~\eqref{c-eq:partial}. All coefficients are bounded and piecewise continuous on finite intervals.
The same argument applies to more general positive self-adjoint
dissipative parts when their quadratic forms are controlled.

\begin{theorem}[Continuous energy and coherent errors]\label{thm:continuous}
For $v=\norm{H_L}_2^2>0$, set
\begin{equation}
 \B(t)=\frac12\int_0^t\ip{H_L}{\mathcal D_sH_L}\,ds,
 \qquad A(t)=\int_0^t\norm{\mathcal A_sH_L}_2\,ds.
\end{equation}
Then
\begin{equation}\label{eq:continuous}
 C_\ell(t)\ge
 \frac{v-\B(t)-2\sqrt v A(t)}{1+\B(t)/v}-E_\ell
 \ge v-2\B(t)-2\sqrt v A(t)-E_\ell.
\end{equation}
For the reset dissipator
$\mathcal D_s=\sum_j\varepsilon_j(s)(\id-P_j)$, suppose each projection
channel $P_j$ has buffer $r_j$. If the Hamiltonian is
$K_s=K_s^{(0)}+\sum_X V_X(s)$, with
$[K_s^{(0)},H_L]=0$ and each $V_X$ having buffer $r_X$, then
\begin{equation}\label{eq:profiles}
 \begin{split}
 \B(t)&\le\tfrac12\int_0^t\sum_j\varepsilon_j(s)\delta_{r_j}^2\,ds,\\
 A(t)&\le2\int_0^t\sum_X\norm{V_X(s)}\delta_{r_X}\,ds.
 \end{split}
\end{equation}
\end{theorem}
\begin{proof}
Evolve $h=H_L/\sqrt v$ by~\eqref{eq:lindblad}, obtaining $\psi_s$, with $\|h\|_2=1$.
Since $\frac d{ds}\norm{\psi_s}_2^2=2\,\mathrm{Re}\ip{\psi_s}{(\mathcal A_s-\mathcal D_s)\psi_s}=-2\ip{\psi_s}{\mathcal D_s\psi_s}$, the skew-adjoint part contributing an imaginary number, the quantity $d=1-\norm{\psi_t}_2^2$ satisfies
$d=2\int_0^t\norm{\mathcal D_s^{1/2}\psi_s}_2^2ds$, where $\mathcal D_s^{1/2}$ is the positive square root and $\norm{\mathcal D_s^{1/2}\psi}_2^2=\ip\psi{\mathcal D_s\psi}$.
Put $c=\ip h{\psi_t}\in\mathbb R$, $y=\B/v$ and $\eta=A/\sqrt v$;
the overlap is real because the propagator preserves Hermiticity.
Integrating the overlap gives
\begin{equation}
 1-c=\int_0^t\langle\mathcal A_s h,\psi_s\rangle_\tau\,ds
 +\int_0^t\langle\mathcal D_s^{1/2}h,
                       \mathcal D_s^{1/2}\psi_s\rangle_\tau\,ds.
\end{equation}
The first integral has absolute value at most $\eta$, because
$\|\psi_s\|_2\le1$. Cauchy--Schwarz in $L^2([0,t])$ bounds the second by
$\big(\int_0^t\ip h{\mathcal D_sh}ds\big)^{1/2}\big(\int_0^t\ip{\psi_s}{\mathcal D_s\psi_s}ds\big)^{1/2}=\sqrt{(2y)(d/2)}$; this is where the factor $\tfrac12$ in $\B(t)$ enters. Therefore
\begin{equation}
 1-c\le\eta+\sqrt{yd}
 \le\eta+\tfrac12\big((1-c)+y(1+c)\big),
\end{equation}
where $d\le1-c^2$ follows from
$|c|\le\|\psi_t\|_2$. Hence
$c\ge(1-y-2\eta)/(1+y)$.
Multiplying by $v$ and subtracting the endpoint error $E_\ell$ of Section~\ref{c-sec:setting} gives the first inequality in~\eqref{eq:continuous}; the second follows from $(1-y-2\eta)/(1+y)-(1-2y-2\eta)=2y(y+\eta)/(1+y)\ge0$.
For~\eqref{eq:profiles}, $\ip{H_L}{(\id-P_j)H_L}=\|(\id-P_j)H_L\|_2^2\le\delta_{r_j}^2$ and $\|\mathcal A_sH_L\|_2\le\sum_X\|[V_X,H_L]\|_2$ by Lemma~\ref{lem:spatial}.
\end{proof}

These estimates weaken exact product conservation in two ways. The reference
dynamics need only conserve $H_L$, not the full product. In particular,
preserving the normalized Gram matrix
$G_L^TG_L/\Tr(G_L^TG_L)$ (the matrix of inner products of the columns of $G_L$, rescaled to unit trace) suffices for the singular-direction charge,
even if $G_L$ changes, because $Z(G_L)$ is the top eigenprojector of $G_L^TG_L$ and is unchanged by rescaling. This also holds for the normalized-Gram charge
introduced below, but not in general for the fixed-image charge. Other
terms may violate this conservation, with their effect controlled by
\eqref{eq:profiles}. This is an approximate-charge statement, not protection
against arbitrary strong noise at the observed edge.

\subsection{Exact conservation with an arbitrary readout}
\label{c-sec:extensions}
We return to the general setting of Section~\ref{c-sec:setting}: an arbitrary Hermitian reference $H$, an arbitrary Hermitian readout $F$ with the parameters $v,w,a$ of~\eqref{c-eq:parameters}, and exact conservation. Consider the time-dependent generator
\begin{equation}
 \frac{dX}{dt}=-i[K(t),X]
 -\sum_j\gamma_j(t)(\id-P_j(t))X,
 \qquad [K(t),H]=0,
 \label{c-eq:generator}
\end{equation}
where $\gamma_j\ge0$, each $P_j$ satisfies the projection-channel
hypotheses, and the coefficients are bounded and piecewise continuous
on finite intervals. We set $\hbar=1$. Each dissipative term generates
a valid unital channel: for fixed $P$,
$e^{t\gamma(P-\id)}=e^{-\gamma t}\id+(1-e^{-\gamma t})P$.
This is a special case of quantum Markov evolution~\cite{GKS,Lindblad,Davies1979}.

\begin{corollary}[Continuous-time correlation bound]
\label{c-cor:continuous}
For the propagator $\Phi_t$ of~\eqref{c-eq:generator}, define
\begin{equation}
 \B(t)=\frac12\int_0^t\sum_j\gamma_j(s)
             \norm{(\id-P_j(s))H}_2^2\,ds.
 \label{c-eq:continuousbudget}
\end{equation}
Then $\ip F{\Phi_t(F)}\ge2a^2/[v+\B(t)]-w$.
\end{corollary}
\begin{proof}
Set $X(s)=\Phi_s(F)$ and $R_j(s)=\id-P_j(s)$.
The Hamiltonian part preserves the norm and the overlap with $H$, while
\begin{align}
 \frac d{ds}\norm{X(s)}_2^2
 &=-2\sum_j\gamma_j(s)\norm{R_j(s)X(s)}_2^2,\\
 a-\ip H{X(t)}
 &=\int_0^t\sum_j\gamma_j(s)
       \ip{R_j(s)H}{R_j(s)X(s)}\,ds.
\end{align}
Cauchy--Schwarz over time and the index $j$ gives
\begin{equation}
 |a-\ip H{X(t)}|^2
 \le\B(t)\,[w-\norm{X(t)}_2^2].
\end{equation}
The geometric part of the proof of Theorem~\ref{thm:overlapreset} now applies
unchanged. The factor $1/2$ in~\eqref{c-eq:continuousbudget} follows
from the factor $2$ in the norm-loss identity; it also agrees with
$p/(2-p)=p/2+O(p^2)$ for infinitesimal partial resets.
\end{proof}

With $F=H$ and $\gamma_j=\varepsilon_j$, the corollary reproduces Theorem~\ref{thm:continuous} with $A=0$; it adds an arbitrary readout at the price of exact conservation.

\subsection{A bound for an imperfect implementation}
The preceding statements concern ideal channels. A separate estimate
is needed to transfer them to a device with imperfect conservation,
nonunital noise, preparation errors, or readout errors.
Let $\widetilde\Phi$ be an implemented channel and suppose
\begin{equation}
 \norm{\widetilde\Phi-\Phi}_\diamond\le\varepsilon.
 \label{c-eq:diamond}
\end{equation}
We use the diamond norm without the conventional extra factor $1/2$;
$\norm\cdot_{\mathrm{tr}}$ below is the unnormalized trace norm.
For a linear map $\Delta$, the diamond norm is the largest induced
trace-norm change after adjoining a reference system:
$\|\Delta\|_\diamond=\sup_{X\ne0}
\|(\Delta\otimes\id_N)(X)\|_{\mathrm{tr}}/\|X\|_{\mathrm{tr}}$.
The reference may be chosen to have dimension $N$.
By the trace duality inequality and the induced trace-norm bound
\cite{Watrous},
\begin{align}
 |\ip F{(\widetilde\Phi-\Phi)(F)}|
 &\le\frac{\norm F_\infty}{N}
           \norm{(\widetilde\Phi-\Phi)(F)}_{\mathrm{tr}}\nonumber\\
 &\le\varepsilon\,
           \frac{\norm F_\infty\norm F_{\mathrm{tr}}}{N}
 \le\varepsilon\norm F_\infty^2.
 \label{c-eq:robust}
\end{align}
Consequently,
\begin{equation}
 \ip F{\widetilde\Phi(F)}\ge
 \frac{2a^2}{v+\B}-w-\varepsilon\norm F_\infty^2.
 \label{c-eq:hardwarebound}
\end{equation}
For a circuit whose ideal and actual factors are $\Phi_k$ and
$\widetilde\Phi_k$, respectively, telescoping the two compositions
gives $\varepsilon\le\sum_k\varepsilon_k$ if
$\norm{\widetilde\Phi_k-\Phi_k}_\diamond\le\varepsilon_k$;
the other channel factors have diamond norm one. This bound can be
conservative, but does not assume that noise only decreases a signal.

For the two-qubit experiment of Section~\ref{c-sec:twoqubit}, whose readout is $Z=I_A\otimes\sigma_z$ and whose score is the measured sign, suppose the mean trace-norm
preparation error is at most $\varepsilon_{\rm prep}$ and the
implemented readout observable differs from the ideal sign by at most
$\varepsilon_{\rm meas}$ in operator norm. Trace-norm contractivity
and $\norm Z_\infty=1$ bound the additional change in the mean score
by $\varepsilon_{\rm prep}+\varepsilon_{\rm meas}$. Thus certified
channel and endpoint errors would give an explicit correction to the
ideal prediction. The experiment reported here does not obtain such
certificates. In particular, a measured echo error is not a bound on
the diamond distance of the reset circuit.

\section{The elementary-matrix family and an explicit signal bound}
Consider
\begin{equation}\label{eq:alphabet}
 \A_m=\{I\}\cup\{I\pm E_{ij}:1\le i\ne j\le m\},
 \qquad D=1+2m(m-1).
\end{equation}
Here $E_{ij}$ is the matrix unit with a single $1$ in entry $(i,j)$, so $E_{ij}^2=0$, $(I\pm E_{ij})^{-1}=I\mp E_{ij}$ and $(I\pm E_{ij})^T=I\pm E_{ji}$; the $2m(m-1)$ nonidentity letters (twelve when $m=3$) $I\pm E_{ij}$ are called transvections or elementary shears. The alphabet is therefore inverse and transpose closed, and it generates
$\SL(m,\mathbb Z)$. Strong irreducibility follows directly. Because the alphabet is inverse closed, the semigroup is the group $\SL(m,\mathbb Z)$. A finite
invariant union of proper subspaces is permuted by the group, so each member has a finite-index
stabilizer, and by pigeonhole some positive power $I+nE_{ij}=(I+E_{ij})^n$ lies in the stabilizer of a chosen
nonzero member $V$ for every $i\ne j$. Then $E_{ij}V\subseteq V$; picking $x\in V$ with $x_j\ne0$ gives $e_i\in V$ for all $i\ne j$, hence $V$ is the whole space, a contradiction. Proximality is witnessed by
\begin{equation}
 (I+E_{12})(I+E_{21})=
 \begin{pmatrix}2&1\\1&1\end{pmatrix}\oplus I_{m-2},
\end{equation}
with eigenvalues $(3\pm\sqrt5)/2$ and, when $m>2$, $1$ of multiplicity
$m-2$. Transpose closure verifies the second set of hypotheses.

\begin{proposition}[Explicit stationary variance]\label{prop:variance}
For~\eqref{eq:alphabet}, every diagonal coordinate of the limiting
projector law has variance
\begin{equation}\label{eq:v0}
 v_*\ge v_{\min}:=\frac{1}{4m^2(D+8)}.
\end{equation}
In particular $v_{\min}=1/756$ for $m=3$.
\end{proposition}
\begin{proof}
The stationary projective measure is unique, because by~\eqref{eq:common} the law of $M_1\cdots M_nx$ converges to that of $\Lambda$ for every input, stationary or not; and conjugation by a coordinate permutation maps the alphabet to itself and preserves the uniform letter law. Together these make its projector law
permutation invariant, so $(q_1,\ldots,q_m)$ is exchangeable with sum one. Let $q_i=x_i^2$ for a unit representative $x$.
Then $\E q_i=1/m$ and $\Var q_i=v$ for all $i$.
On
\begin{equation}
 \Omega=\{|q_1-1/m|\le1/(4m),\ |q_2-1/m|\le1/(4m)\},
\end{equation}
one of $I+E_{12}$ and $I-E_{12}$, which map $x$ to $x\pm x_2e_1$ and change only the first coordinate, changes $x_1$ to a number of magnitude
$||x_1|-|x_2||=|\sqrt{q_1}-\sqrt{q_2}|$. Its square is at most $1/(12m)$,
because $|q_1-q_2|\le1/(2m)$ and
$\sqrt{q_1}+\sqrt{q_2}\ge\sqrt{3/m}$.
After normalization the new first squared coordinate is at most
\begin{equation}
 \frac{1/(12m)}{1-5/(4m)}
 =\frac{1}{3(4m-5)}\le\frac{1}{2m}.
\end{equation}
Chebyshev's inequality gives $\Prob(\Omega)\ge1-32m^2v$.
The stationarity equation for the first-coordinate variance reads $v=D^{-1}\sum_{a\in\A_m}\E\,(q_1(a\cdot x)-1/m)^2$ with $a\cdot x=ax/\|ax\|$ (stationarity of Section~\ref{sec:localization} applied to the function $(q_1-1/m)^2$). On $\Omega$, at
least one of those two letters gives $q_1'\le1/(2m)$, hence $|q_1'-1/m|\ge1/(2m)$ and a squared deviation of at least
$1/(4m^2)$. All other terms are nonnegative. Therefore
\begin{equation}
 v\ge\frac{1-32m^2v}{4Dm^2},
\end{equation}
which rearranges to~\eqref{eq:v0}. Both letters occur in the transpose
law, so the stationarity equation is the required one.
\end{proof}

This numerical signal bound does not assign numerical values to
$\kappa$ or $C_{\rm loc}$ for the singular-direction charge, the charge $H_L$ of~\eqref{eq:charges} (so named in Table~\ref{tab:charges}). Those depend on the random-product estimates.
The alternative charge of Appendices~\ref{sec:explicit}--\ref{sec:sharper}
has explicit constants and a stronger variance bound when $m=3$.
Proposition~\ref{prop:variance} does, however, give an explicit sufficient readout size in terms of them:
\begin{equation}\label{eq:ell0}
 \ell_0=\max\left\{1,
 \left\lceil\kappa^{-1}\log\frac{8C_{\rm loc}}{v_{\min}}\right\rceil\right\}.
\end{equation}
For $L\ge\ell\ge\ell_0$, $d_{L,\ell}\le\delta_\ell\le v_{\min}/8$. Taking $r=L$ in~\eqref{eq:varfinite} gives $|\sqrt{v_L}-\sqrt{v_*}|\le v_{\min}/8$, and $\sqrt{v_L}+\sqrt{v_*}\le1$ then gives $|v_L-v_*|\le v_{\min}/8$, so $v_L\ge7v_{\min}/8\ge3v_{\min}/4$.

Below, $P$ denotes the complete reset~\eqref{c-eq:reset} of the first $L-r$ sites, $t$ counts rounds, and $C_\ell(t)$ is the autocorrelation defined before Theorem~\ref{thm:linear}, with $\Phi_t$ the composition of the first $t$ rounds.

\begin{corollary}[Macroscopic noisy prefix]\label{cor:prefix}
Take $L\ge\ell\ge\ell_0$. Between arbitrary $H_L$-preserving
bistochastic channels, apply one partial complete reset per round to the first $L-r$
sites using the channel $(1-p)\id+pP$, $p\in(0,1]$. There is no restriction on the
width $L-r$ except $1\le r<L$. Then
\begin{equation}\label{eq:bufferwindow}
 C_\ell(t)\ge v_{\min}/2\quad\hbox{for integer }t\ge0\hbox{ with}\quad
 t\le \frac{v_{\min}(2-p)}{16pC_{\rm loc}^2}e^{2\kappa r}.
\end{equation}
For $r\ge\ell$ the reset is disjoint from the readout block. Taking
$r$ proportional to $L$ allows an extensive noisy region while retaining
an exponential-in-$L$ lower window. An arbitrary bistochastic channel on
that prefix instead gives the weaker exponential-in-$r$ window from
Theorem~\ref{thm:linear}.
\end{corollary}
\begin{proof}
The budget $\B$ of~\eqref{c-eq:parameters} is at most
$tpC_{\rm loc}^2e^{-2\kappa r}/(2-p)$, since each of the $t$ resets has $b_j\le\delta_r$ by Lemma~\ref{lem:spatial}.
Under~\eqref{eq:bufferwindow} it is at most $v_{\min}/16$.
Use~\eqref{eq:reset} with $v_L\ge3v_{\min}/4$ and $E_\ell\le d_{L,\ell}\le v_{\min}/8$: $3/4-2/16-1/8=1/2$.
\end{proof}

The preparation $\rho_{0,L}=D^{-L}(I+\zeta F_\ell)$, $0<\zeta\le1$,
is positive and differs from infinite temperature only in its edge block.
Its signal is $\Tr(F_\ell\rho_t)=\zeta C_\ell(t)$, by unitality and $\Tr F_\ell=0$. Let $\rho_{t,\mathrm{edge}}$ be the reduced state of the last $\ell$ sites and write $F_\ell=I\otimes f_\ell$ with $\|f_\ell\|\le1$. Throughout
\eqref{eq:bufferwindow}, trace-norm duality, $\zeta C_\ell(t)=\Tr(f_\ell(\rho_{t,\mathrm{edge}}-I/D^\ell))\le\|\rho_{t,\mathrm{edge}}-I/D^\ell\|_{\mathrm{tr}}$, yields
\begin{equation}
 \tfrac12\norm{\rho_{t,\mathrm{edge}}-I/D^\ell}_{\mathrm{tr}}
 \ge\zeta v_{\min}/4.
\end{equation}
The diagonal charge stores a classical signal in quantum dynamics;
no protection of an arbitrary unknown qubit is asserted.

\section{Noise-limited lifetime and a sharp scaling obstruction}
Take the alphabet~\eqref{eq:alphabet}, so that $v_{\min}$ and $\ell_0$ of Proposition~\ref{prop:variance} and~\eqref{eq:ell0} apply, an $H_L$-preserving Hamiltonian, a prefix reset of rate $\Gamma$
with buffer $r$, and independent single-site reset noise of rate
$\varepsilon$ on every site; a reset of rate $\gamma$ is the generator $\gamma(P-\id)$ of Section~\ref{c-sec:extensions}, and single-site reset noise is complete depolarization $P_k$ of site $k$. The dissipator is $\mathcal D=\Gamma(\id-P_{\rm pre})+\varepsilon\sum_{k=1}^L(\id-P_k)$. The reset at the rightmost site
has buffer zero; successive sites have buffers $1,2,\ldots,L-1$.
Using $\ip{H_L}{(\id-P)H_L}=\|(\id-P)H_L\|_2^2\le\delta_{\rm buffer}^2$ (Lemma~\ref{lem:spatial}), the bound $\|(\id-P_L)H_L\|_2^2\le v_L\le1/4$ for the rightmost-site reset and the geometric series $\sum_{k\ge1}C_{\rm loc}^2e^{-2\kappa k}$, the dissipative quadratic form is bounded, uniformly in $L$, by
\begin{equation}
 \ip{H_L}{\mathcal D H_L}\le
 \Gamma C_{\rm loc}^2e^{-2\kappa r}+\varepsilon\Xi,
 \qquad
 \Xi=\frac14+\frac{C_{\rm loc}^2e^{-2\kappa}}{1-e^{-2\kappa}}.
\end{equation}
With $A(t)=0$ and $\B(t)=\tfrac t2\ip{H_L}{\mathcal DH_L}$, and with $v_L\ge3v_{\min}/4$ and $E_\ell\le v_{\min}/8$ for $L\ge\ell\ge\ell_0$ (the discussion following~\eqref{eq:ell0}), the second form of~\eqref{eq:continuous} therefore proves
\begin{equation}\label{eq:crossover}
 C_\ell(t)\ge v_{\min}/2\quad\hbox{for}\quad
 0\le t\le
 \frac{v_{\min}}{8\left(\Gamma C_{\rm loc}^2e^{-2\kappa r}
                       +\varepsilon\Xi\right)},
\end{equation}
provided $L\ge\ell\ge\ell_0$.
This is a sufficient lower window with constants independent of $L$.
The noise contribution is $O(\varepsilon)$, rather than the
$O(L\varepsilon)$ obtained by ignoring charge localization and bounding each site term by $\|H_L\|_2^2\le1/4$. The window crosses from buffer-limited to noise-limited at $\varepsilon\approx\Gamma C_{\rm loc}^2e^{-2\kappa r}/\Xi$.
It crosses from an exponential buffer-limited window to an inverse-noise
window. Spatially nonuniform rates are treated by the actual weighted sum
$\sum_j\varepsilon_j\delta_{r_j}^2$.

A bistochastic Markov generator is one whose semigroup $e^{s\mathcal G}$, $s\ge0$, consists of bistochastic channels; the only property used is that it does not increase $\|\cdot\|_2$.

\begin{proposition}[Uniform depolarization obstruction]\label{prop:upper}
Let the continuous generator contain independent single-site complete
depolarization of rate $\varepsilon>0$ at every site. Allow any
time-dependent Hamiltonian and additional bistochastic Markov generators.
For every centered Hermitian $F$,
\begin{equation}\label{eq:upper}
 |\tau_L(F\Phi_t(F))|\le e^{-\varepsilon t}\norm F_2^2.
\end{equation}
\end{proposition}
\begin{proof}
Take a $\tau$-orthonormal single-site operator basis containing $I$ and form the tensor-product strings $\sigma_{a_1}\otimes\cdots\otimes\sigma_{a_L}$, an orthonormal basis of operators. $P_j$ fixes a string whose $j$-th factor is $I$ and annihilates the others, so $\id-P_j$ is the orthogonal projection onto strings with a nonidentity factor at site $j$. Each traceless string
has at least one nonidentity factor. Therefore, by Parseval, on the traceless subspace
\begin{equation}
 \sum_{j=1}^L\norm{(\id-P_j)X}_2^2\ge\norm X_2^2.
\end{equation}
The trace-preserving propagator keeps $\Phi_t(F)$ traceless, so this applies at every time.
The Hamiltonian contributes zero to the derivative of $\norm X_2^2$.
Each additional bistochastic generator contributes a nonpositive amount.
Thus $d\norm{\Phi_t(F)}_2^2/dt\le
-2\varepsilon\norm{\Phi_t(F)}_2^2$. Equivalently,
$e^{2\varepsilon t}\|\Phi_t(F)\|_2^2$ is nonincreasing, whence
$\|\Phi_t(F)\|_2\le e^{-\varepsilon t}\|F\|_2$.
Cauchy--Schwarz proves~\eqref{eq:upper}.
\end{proof}

Equality in~\eqref{eq:upper} holds for a traceless one-site observable with
only uniform depolarization and no Hamiltonian: then $\sum_j(P_j-\id)F=-F$, so $\Phi_t(F)=e^{-\varepsilon t}F$. A fixed fraction surviving
to time $e^{cL}$ requires $\varepsilon=O(e^{-cL})$.
This obstruction rules out a constant-fraction memory plateau lasting
exponentially in $L$ at fixed homogeneous depolarization rate within this
model. It does not address active error correction with fresh ancillas or
other dynamics outside the bistochastic assumptions. Coherent errors are
instead controlled by $A(t)$ in~\eqref{eq:profiles}; no matching general
upper bound is claimed for them.

\section{Explicit constants and the distinction between charges}
\label{sec:explicit-summary}
For $m=3$, the alphabet~\eqref{eq:alphabet} consists of thirteen symbols.
Each physical site therefore has thirteen basis states; the $3\times3$
matrices are labels of those states. There are three useful choices of
scalar charge, all centered by subtracting their uniform-word means.
They must be distinguished because their conservation requirements and
certified constants differ.

\begin{table}[htb]
\centering\small
\begin{tabular}{@{}p{0.24\linewidth}p{0.36\linewidth}p{0.31\linewidth}@{}}
\toprule
Charge & Uncentered value & Certified localization and signal\\
\midrule
Singular-direction & $Z(G)_{jj}=v_j^2$, with $v$ the unit right singular vector of the largest singular value of $G$ (Section~\ref{sec:localization})
  & $C_{\rm loc}e^{-\kappa r}$; $v_*\ge1/19$ for $j=1$ (Appendix~\ref{sec:sharper})\\[4pt]
Fixed-image & $G_{11}^2/(G_{11}^2+G_{12}^2+G_{13}^2)$
  & $e^{-r/1400}$; $v_*\ge1/19$\\[4pt]
Normalized-Gram & $\|Ge_1\|^2/\|G\|_F^2$
  & $3e^{-r/1400}$; $v_*\ge1/19$\\
\bottomrule
\end{tabular}
\caption{Bounds on $\|H_L-F_r\|_2$ and the limiting variance $v_*$ (the signal) for three
different charges; for each charge, $H_L$ denotes its centered full-chain operator and $F_r$ the same charge computed from the last $r$ sites. $C_{\rm loc}$ and $\kappa$ are the non-numerical constants of Theorem~\ref{thm:local}. The explicit localization constants in the last two rows and the variance bound $1/19$ in all three rows follow from exact polynomial certificates and rational estimates in Appendices~\ref{sec:explicit} and~\ref{sec:sharper}, with the common stationary law identified in Appendix~\ref{sec:explicit}; they are not fitted to trajectories.}
\label{tab:charges}
\end{table}

The fixed-image charge is the first squared coordinate of the normalized
vector $G^Te_1$. It avoids computing a singular vector and is rational on
integer matrices. Appendix~\ref{sec:explicit} first proves a weaker
$e^{-r/10000}$ bound, laying out the projective contraction calculation: the estimate showing that $r$ random letters pull two input directions together, in the sine of their angle, at an exponential rate.
Appendix~\ref{sec:sharper} strengthens the same calculation using the
orthogonality of two vectors and a sharper stationary-variance estimate.
The improvement is analytic and computer-assisted: every positivity
certificate (a polynomial proved nonnegative by rewriting it, after a change of variables, with all coefficients nonnegative) is checked with integer or rational arithmetic.

For the fixed-image charge, a suffix readout of at least
$\lceil700\log190\rceil=3673$ sites is sufficient for correlation
$C_\ell(t)\ge1/38$ whenever
\begin{equation}
 t\le\frac{2-p}{304p}e^{r/700},\qquad t\in\mathbb Z_{\ge0}.
 \label{eq:summarywindow}
\end{equation}
Here a round is one $\widetilde H_L$-preserving bistochastic channel together with one prefix reset $(1-p)\id+pP$, where $P$ replaces the first $L-r$ sites by the maximally mixed state and the $r$ untouched sites are the buffer; $t$ counts rounds, the readout is $\widetilde F_\ell$, $C_\ell(t)=\tau_L[\widetilde F_\ell\Phi_t(\widetilde F_\ell)]$, and $L\ge\ell$, $1\le r<L$ (Corollary~\ref{cor:sharperwindow}). The reset region and the readout block may overlap. The margins are $190=10/\nu$ and $304=16/\nu$ with $\nu=1/19$, and $1/38=\nu/2$.
For $p=1$, \eqref{eq:summarywindow} reads $t\le e^{r/700}/304$, so $t$ rounds are certified once $r\ge700\log(304t)$: one complete round by the buffer
$\lceil700\log304\rceil=4002$, a thousand rounds by
$\lceil700\log304000\rceil=8838$. These are upper bounds on the sizes
needed by this proof, not lower bounds on a device's necessary size.

The normalized-Gram charge is preserved whenever $G^TG/\Tr(G^TG)$ is
preserved, a weaker requirement than conserving $G$ itself: $G^TG$ is the Gram matrix of inner products of the columns of $G$, and $u_{\rm G}$ is the $(1,1)$ entry of its unit-trace rescaling. Left multiplication of $G$ by a rotation keeps $G^TG$ but changes the first row of $G$, which is why the fixed-image charge is not preserved in general. The normalized-Gram
localization proof and asymptotic window follow next. Direct finite-chain
averaging (exact enumeration of all $13^8$ words, aggregated by symmetry) then gives a much smaller example: for eight thirteen-state sites, the readout $F=PH_L^{\rm G}$ of~\eqref{eq:conditionalgram}, the charge averaged over the thirteen possible first letters and hence a function of the last seven sites, retains $C_F(2)>\frac59\|F\|_2^2$ after two complete first-site resets~\eqref{eq:eightsite}.
This does not use the large asymptotic readout bounds and has not been implemented on hardware.

\section{Normalized-Gram conservation and finite-chain certificates}\label{sec:gram}
The fixed-image charge of Table~\ref{tab:charges} selects the input direction $e_1$ before applying the
matrix product. Averaging the analogous quantities $G_{i1}^2/\|G^Te_i\|^2$ of the three rows of $G$ with the weights $\|G^Te_i\|^2/\|G\|_F^2$ gives a different
rational charge,
\begin{equation}\label{eq:gramcharge}
 u_{\rm G}(G)=\frac{\norm{Ge_1}^2}{\norm G_F^2}
 =\left(\frac{G^TG}{\Tr(G^TG)}\right)_{11}.
\end{equation}
Let $H_L^{\rm G}=\diag(u_{\rm G}(G_L)-\tfrac13)$ and $F_r^{\rm G}=\diag(u_{\rm G}(B_r)-\tfrac13)$, extended by the identity outside the last $r$ sites, be its centered full-word and
suffix versions; the mean of $u_{\rm G}$ is exactly $1/3$ at every length, by the permutation symmetry recalled below. Conservation of the normalized Gram matrix suffices
to preserve $H_L^{\rm G}$, even when the full product changes, because $u_{\rm G}$ is a function of $G^TG/\Tr(G^TG)$ alone.
For quantum dynamics this means that the channel fixes the corresponding
diagonal charge; unitaries block diagonal in normalized-Gram fibres are
one example. Conservation of only an ensemble average is insufficient.

\begin{proposition}[Explicit normalized-Gram localization]\label{prop:gramlocal}
For the uniform thirteen-letter alphabet and $L\ge r\ge1$,
\begin{equation}
 \norm{H_L^{\rm G}-F_r^{\rm G}}_2\le3e^{-r/1400},\qquad
 |\norm{F_n^{\rm G}}_2^2-v_*|\le6e^{-n/700},\qquad v_*\ge1/19.
\end{equation}
Apply one partial complete reset per round of strength $p\in(0,1]$ on the first $L-r$ sites, interleaved with the conserving channels specified below. Then $L\ge\ell\ge5211$ suffices for
\begin{equation}\label{eq:gramwindow}
 C_\ell^{\rm G}(t)\ge\frac1{38},\qquad
 0\le t\le\frac{2-p}{2736p}e^{r/700},\quad t\in\mathbb Z,\quad 1\le r<L.
\end{equation}
The readout is $F_\ell^{\rm G}$ and
$C_\ell^{\rm G}(t)=\tau_L[F_\ell^{\rm G}\Phi_t(F_\ell^{\rm G})]$,
where $\Phi_t$ contains $t$ reset factors interleaved, in any prescribed
order, with $H_L^{\rm G}$-preserving bistochastic channels.
\end{proposition}
\begin{proof}
Write $c(v)=v_1^2/\norm v^2$ and
\begin{equation}
 u_{\rm G}(G)=\sum_{i=1}^3\alpha_i(G)c(G^Te_i),\qquad
 \alpha_i(G)=\frac{\norm{G^Te_i}^2}{\norm G_F^2}.
\end{equation}
For $G=AB_r$, subtract the analogous expression for $B_r$, with
weights $\beta_j=\alpha_j(B_r)$. The difference is
\begin{equation}
 \sum_{i,j}\alpha_i(AB_r)\beta_j
 \left[c(B_r^TA^Te_i)-c(B_r^Te_j)\right].
\end{equation}
These weights depend on the suffix. Jensen's inequality is applied
pointwise first; only then bound each weight by one. Conditional on
$A$, each of the nine input pairs is a fixed pair of nonzero vectors independent of the $r$ suffix letters, so the iterated one-letter contraction certificate of Appendix~\ref{sec:sharper} (the argument giving~\eqref{eq:sharperraw}), together with $|c(x)-c(y)|\le d([x],[y])$ from Appendix~\ref{sec:explicit} and $d^2\le d\le d^{3/50}$ on $[0,1]$, gives raw squared and first absolute
errors at most $9e^{-r/700}$. Both means equal $1/3$, so centering changes nothing and the square root gives the stated $L^2$ bound.

For the stationary comparison, take a random line $Y$, independent of the letters and distributed by the stationary law of the transpose walk (Section~\ref{sec:localization}), and compare all three row images $G_n^Te_i$ to
$G_n^TY$. The latter has stationary law at every $n$. The same argument uses just
three pairs and gives $\E|u_{\rm G}(G_n)-c(G_n^TY)|\le3e^{-n/700}$.
The variance comparison $|\Var U-\Var V|\le2\E|U-V|$ of Appendix~\ref{sec:sharper} gives the
second estimate and identifies the limiting variance with $\Var\,c(Y)$, the same $v_*\ge1/19$ as in Theorem~\ref{thm:sharper}.

Use the common envelope $\eta_r=9e^{-r/700}$ for the squared
localization error and the raw stationary first-moment error.
Put $\nu=1/19$. Then $\eta_\ell\le\nu/10$ for
$\ell\ge\lceil700\log1710\rceil=5211$.
For $v=\|H_L^{\rm G}\|_2^2$, $w=\|F_\ell^{\rm G}\|_2^2$ and
$a=\ip{H_L^{\rm G}}{F_\ell^{\rm G}}$, the estimates above imply
$v,w\ge\nu-2\eta_\ell$ and $a\ge(v+w-\eta_\ell)/2$.
Each reset fixes $F_r^{\rm G}$, so the proof of Lemma~\ref{lem:spatial}
gives $b_j^2\le9e^{-r/700}$ and
\begin{equation}
 \B\le\frac{9tp}{2-p}e^{-r/700}\le\frac\nu{16}
 \quad\text{whenever}\quad
 t\le\frac{2-p}{(16\cdot19\cdot9)p}e^{r/700}.
\end{equation}
Applying Theorem~\ref{thm:overlapreset}, exactly as in the proof of
Corollary~\ref{cor:sharperwindow}, gives
$C_\ell^{\rm G}(t)\ge58\nu/115=58/2185>1/38$.
This proves~\eqref{eq:gramwindow} and explains the constant $2736$.
For one complete reset a sufficient buffer is
$\lceil700\log2736\rceil=5540$; the minimal readout buffer alone
again does not certify a complete round.
\end{proof}

\subsection{A finite readout and the final reset}
For a specified finite $L$, let $P$ be the complete reset~\eqref{c-eq:reset} of the first site, which on a diagonal function of the word averages the first letter; since $G_L=aB$ with $B=a_2\cdots a_L$, the readout $F=PH_L^{\rm G}$ is
\begin{equation}\label{eq:conditionalgram}
 F=PH_L^{\rm G},\qquad
 F(B)=\frac1{13}\sum_{a\in\A_3}u_{\rm G}(aB)-\frac13.
\end{equation}
It is supported on the last $L-1$ sites and may depend on $L$.
Permutation symmetry gives the exact finite-word mean $1/3$.
Define
\begin{equation}
 v=\norm{H_L^{\rm G}}_2^2,\quad w=\norm F_2^2,\quad
 b^2=\norm{(I-P)H_L^{\rm G}}_2^2=v-w,\quad\lambda=b^2/v.
\end{equation}
Orthogonal projection gives $\ip{H_L^{\rm G}}F=w$.
Among all suffix-supported observables, $F$ maximizes the squared
normalized overlap with this particular charge, by Cauchy--Schwarz.

Let every round be $Q V_j$, where $V_j$ is a bistochastic channel
fixing $H_L^{\rm G}$ and $Q=(1-p)I+pP$.
Since $P=P^*$ and $PF=P^2H_L^{\rm G}=F$, one has $Q^*F=F$, so the final reset drops out of the autocorrelation $C_F(t)=\ip F{QV_t\cdots QV_1F}$. The remaining product contains $t-1$ resets, each with $b_j^2=b^2$, so $\B=(t-1)pb^2/(2-p)$; with $a=\ip{H_L^{\rm G}}F=w$, Theorem~\ref{thm:overlapreset} gives $C_F(t)\ge2w^2/(v+\B)-w$, and dividing by $w$ with $w/v=1-\lambda$ yields
\begin{equation}\label{eq:finitegram}
 \frac{C_F(t)}w\ge
 \frac{2(1-\lambda)}{1+(t-1)p\lambda/(2-p)}-1,
 \qquad t\ge1,\quad w>0.
\end{equation}
All $t-1$ interior resets remain in the budget. In particular, a
half-variance plateau, meaning $C_F(t)\ge\|F\|_2^2/2$, follows whenever
\begin{equation}\label{eq:finitegramhalf}
 \lambda\le\frac{1}{4+3(t-1)p/(2-p)}.
\end{equation}
For two complete-reset rounds the threshold is $\lambda\le1/7$.
This is a finite-chain certificate, distinct from the fixed-support,
arbitrary-length result above.

\subsection{Exact eight-site certificate}
The required finite moments can be evaluated without simulating a
quantum channel. For a suffix product $B$, set
\begin{equation}
 m_1(B)=\frac1{13}\sum_a u_{\rm G}(aB),\qquad
 m_2(B)=\frac1{13}\sum_a u_{\rm G}(aB)^2.
\end{equation}
If $N_{L-1}(B)$ counts suffix words with product $B$, then
\begin{equation}\label{eq:finitemoments}
 v=13^{-(L-1)}\sum_B N_{L-1}(B)m_2(B)-\frac19,\qquad
 w=13^{-(L-1)}\sum_B N_{L-1}(B)m_1(B)^2-\frac19.
\end{equation}
Every term is rational. Exact integer multiplicities can also be
aggregated modulo left signed row permutations. Indeed
$u_{\rm G}(RB)=u_{\rm G}(B)$ for an orthogonal signed permutation $R$,
and conjugation by $R$ permutes the alphabet. Thus the left-product
transition and both moments descend to these orbits, without changing
the sums. Equivalently, one may group integer unimodular products by
their Gram matrix: if $B^TB=C^TC$, then $R=CB^{-1}$ is an integer
orthogonal matrix, hence a signed permutation. A representative of
determinant $-1$ is harmless: it represents
the orbit for this computation and is not an extra physical letter.

The ancillary verifier evaluates~\eqref{eq:finitemoments} exactly
through $L=8$. Rounded displays are
\begin{center}
\begin{tabular}{r r r r}
$L$ & $v$ & $w$ & $\lambda=b^2/v$\\\hline
5 & $0.03956777$ & $0.03093240$ & $0.21824263$\\
6 & $0.04368525$ & $0.03600766$ & $0.17574807$\\
7 & $0.04716739$ & $0.04031135$ & $0.14535540$\\
8 & $0.05013681$ & $0.04398668$ & $0.12266696$\\\hline
\end{tabular}
\end{center}
The certificate uses the exact fractions, not these displays.
In particular $v_8>1/20$ and $8b_8^2<v_8$, hence $\lambda_8<1/8<1/7$.
Thus every allowed sequence
of two complete-reset rounds satisfies
\begin{equation}\label{eq:eightsite}
 C_F(2)>\frac59\norm F_2^2>\frac7{288}.
\end{equation}
The first inequality follows by substituting $\lambda<1/8$ in
\eqref{eq:finitegram}; the second uses $w=(1-\lambda)v>7/160$.
One complete reset is interior to the evolution and is charged in this
bound. As a partial-reset example, at $L=6$, $p=1/10$, the exact
certificate gives $C_F(11)\ge\norm F_2^2/2$ with ten interior resets.

These are finite computer-assisted corollaries of the overlap theorem,
uniform over the stated conserving bulk channels. Each site has thirteen
states; eight sites are not an eight-qubit implementation. This seven-site readout is defined for $L=8$ only and does not shorten the readout of Proposition~\ref{prop:gramlocal}. Being diagonal, $F$
can be evaluated from measured suffix labels using~\eqref{eq:conditionalgram}.
A conserving gate set, an initial-state preparation and their noise
budgets must still be implemented and checked before claiming a hardware
demonstration. No physical simulation or experimental result is inferred
from these moment certificates. A three-active-site hardware pilot of the four-letter $\SL(2,\mathbb Z)$ analogue of this finite certificate is reported in Appendix~\ref{c-app:earlier}; it did not meet its predeclared primary criterion.

\section{An exactly solvable two-qubit reset echo}
\label{c-sec:twoqubit}
This section instantiates Theorem~\ref{thm:overlapreset} in the smallest system that admits a nontrivial complete reset and a separate nonconstant readout; the matrix-product chain is a separate application of the same inequality. In this section $A$, $B$ denote qubits, tensor factors are ordered $A\otimes B$, and $c,s,\theta,\lambda$ are the gate parameters below. Let $A$ be the reset qubit and $B$ the readout qubit. Write
$Z=I_A\otimes\sigma_z$, $M=\sigma_x\otimes\sigma_x$, and
\begin{equation}
 H=cZ+sM,\qquad c=\cos\theta,\quad s=\sin\theta,
 \qquad \lambda=s^2.
 \label{c-eq:Htwo}
\end{equation}
Here $H$ is the conserved reference observable of Section~\ref{c-sec:setting}, not a Hamiltonian and not the Hadamard gate, which is written $\mathsf H$ below; $\lambda$ will turn out to be the reset leakage. The Pauli strings satisfy $Z^2=M^2=I$, $ZM=-MZ$, and are trace
orthogonal, $\tau(ZM)=0$. Hence $H^2=c^2Z^2+s^2M^2+cs(ZM+MZ)=I$, $\tau(H)=0$, and $v=\tau(H^2)=1$.
The gate $V=H$ is a Hermitian unitary ($V^\dagger=V$, $V^2=I$) and exactly conserves $H$, since $V$ commutes with itself.
It is an interacting gate: $V|00\rangle=c|00\rangle+s|11\rangle$.
Define
\begin{equation}
 F=P_AH=cZ,\qquad
 \Phi_\theta=\Ad_{V^\dagger}\circ P_A\circ\Ad_V.
 \label{c-eq:echo}
\end{equation}
Here $\Ad_V(X)=VXV^\dagger$, $P_A$ is the complete reset~\eqref{c-eq:reset} of qubit $A$, and, read right to left, $\Phi_\theta$ applies $V$, resets $A$, then applies $V^\dagger=V$; without the reset it would be the identity.
Since $P_A(I_A\otimes\sigma_z)=(I_A/2)\Tr(I_A)\otimes\sigma_z=Z$ and $P_A(\sigma_x\otimes\sigma_x)=(I_A/2)\Tr(\sigma_x)\otimes\sigma_x=0$, and the reset is complete ($p=1$, coefficient $p/(2-p)=1$),
\begin{equation}
 w=a=c^2=1-\lambda,\qquad
 \B=\norm{(\id-P_A)H}_2^2=\lambda.
\end{equation}
There is one interior complete reset. Theorem~\ref{thm:overlapreset} yields
\begin{equation}
 \frac{C_F}{w}\ge\frac{2(1-\lambda)}{1+\lambda}-1
 =\frac{1-3\lambda}{1+\lambda}.
 \label{c-eq:twobound}
\end{equation}

To compute the correlation exactly, expand the conjugation:
\begin{align}
 VZV&=c^2 ZZZ+cs(ZZM+MZZ)+s^2MZM\nonumber\\
 &=(c^2-s^2)Z+2csM,\\
 P_A(VFV)&=c(c^2-s^2)Z=c(1-2\lambda)Z.
\end{align}
Using cyclicity of the trace, $\ip F{V^\dagger YV}=\ip{VFV^\dagger}Y$ (so $\Ad_V$ is an isometry), and $P_A=P_A^2=P_A^*$ (so $\ip Y{P_AY}=\|P_AY\|_2^2$), with $\|Z\|_2^2=\tau(Z^2)=1$,
\begin{equation}
 C_F=\ip{VFV}{P_A(VFV)}
 =\norm{P_A(VFV)}_2^2
 =c^2(1-2\lambda)^2.
\end{equation}
Therefore
\begin{equation}
 \boxed{\displaystyle\frac{C_F}{w}=(1-2\lambda)^2,\qquad
 (1-2\lambda)^2-\frac{1-3\lambda}{1+\lambda}
 =\frac{4\lambda^3}{1+\lambda}.}
 \label{c-eq:exact}
\end{equation}
The general bound matches the exact result through second order at
small leakage. This establishes tightness to that order for this
family, not optimality of the bound for all channels.

At $\theta=\pi/6$, $\lambda=1/4$, the lower bound is $1/5$ and
the exact correlation is $1/4$. At $\lambda=1/8$, the corresponding
values are $5/9$ and $9/16$, separated by $1/144\approx0.007$, below the sampling half-width $0.025$ of~\eqref{c-eq:CI}; this smaller
margin, against $1/20$ at $\lambda=1/4$, motivated choosing $\lambda=1/4$ for the
hardware run. Figure~\ref{c-fig:theory} displays the dependence on
$\lambda$. The theorem needs a nontrivial subsystem that is completely reset and a disjoint subsystem carrying a nonconstant readout; one qubit for each is the minimum, since a completely reset single qubit leaves only scalars. This is not a minimum dimension for every application of the theorem.

\begin{figure}[tb]
 \centering
 \includegraphics[width=0.88\linewidth]{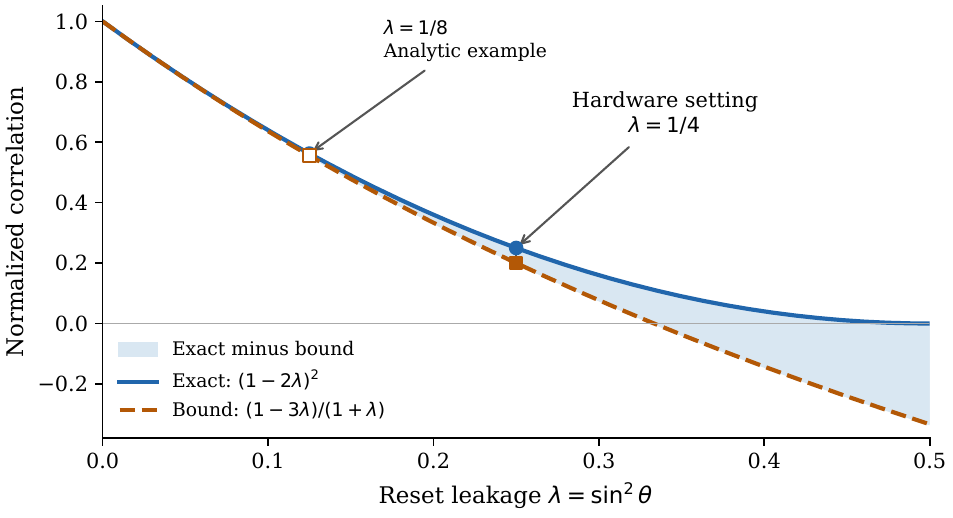}
 \caption{Exact normalized correlation and the general lower bound
 for the two-qubit family. The shaded difference is
 $4\lambda^3/(1+\lambda)$. The marked hardware setting is
 $\lambda=1/4$, where the exact value is $0.25$ and the bound is $0.20$.
 The additional marker at $\lambda=1/8$ is an analytic example; it was
 not a hardware setting. These curves are predictions, not experimental fits.}
 \label{c-fig:theory}
\end{figure}

\section{Hardware implementation and results}
\label{c-sec:experiment}
\subsection{Circuit, reset, and controls}
The conserving gate has the factorization
\begin{equation}
 V=e^{-i\theta\sigma_x\otimes\sigma_y}Z.
 \label{c-eq:gate}
\end{equation}
Indeed, $(\sigma_x\otimes\sigma_y)Z=iM$, so the right-hand side
is $cZ+sM$. To make the entangling-gate count explicit, let $\mathsf H$
be the Hadamard gate, $S_{\rm ph}=\operatorname{diag}(1,i)$ the phase gate, and
$W=\mathsf H\otimes(S_{\rm ph}\mathsf H)$; then $\mathsf H\sigma_z\mathsf H=\sigma_x$ and $S_{\rm ph}\sigma_xS_{\rm ph}^\dagger=\sigma_y$, which is why $W$ maps $\sigma_z\otimes\sigma_z$ to $\sigma_x\otimes\sigma_y$. Then
$W(\sigma_z\otimes\sigma_z)W^\dagger=\sigma_x\otimes\sigma_y$, so
\begin{equation}
 e^{-i\theta\sigma_x\otimes\sigma_y}
 =W\,\mathrm{CNOT}_{A\to B}\,[I\otimes R_z(2\theta)]\,
       \mathrm{CNOT}_{A\to B}\,W^\dagger,
 \quad R_z(\varphi)=e^{-i\varphi\sigma_z/2}.
\end{equation}
The forward and inverse gates therefore use four entangling gates
in total. Native compilation on the selected IBM qubits used four
CZ gates and depth 32.
Circuits were constructed and compiled with Qiskit~\cite{Qiskit};
archived circuit files and runtime metadata specify the submitted template.

The complete reset is implemented by equally weighting the four
Pauli masks on $A$, using equal numbers of shots for each of four Pauli-mask circuits on $A$ (a balanced implementation of the Pauli twirl):
\begin{equation}
 P_A(X)=\frac14\sum_{\mu=0}^3
   (\sigma_\mu\otimes I)X(\sigma_\mu\otimes I),
 \quad (\sigma_0,\sigma_1,\sigma_2,\sigma_3)=(I,\sigma_x,\sigma_y,\sigma_z).
 \label{c-eq:twirl}
\end{equation}
To verify this for arbitrary, possibly entangled inputs, expand
$X=\sum_{\nu=0}^3\sigma_\nu\otimes X_\nu$.
The average conjugation preserves the $\nu=0$ term and cancels
each of the other Pauli terms, leaving
$I_A\otimes X_0=(I_A/2)\otimes\Tr_A X$.
Each individual mask is unitary; the reset is the ensemble channel.
The mask is independent of the input and no outcomes are selected
according to its value (no postselection). For $\nu\ne0$, $\sigma_\mu\sigma_\nu\sigma_\mu=\pm\sigma_\nu$ with two signs of each kind, which is why the four terms cancel.

\begin{figure}[tb]
 \centering
 \includegraphics[width=\linewidth]{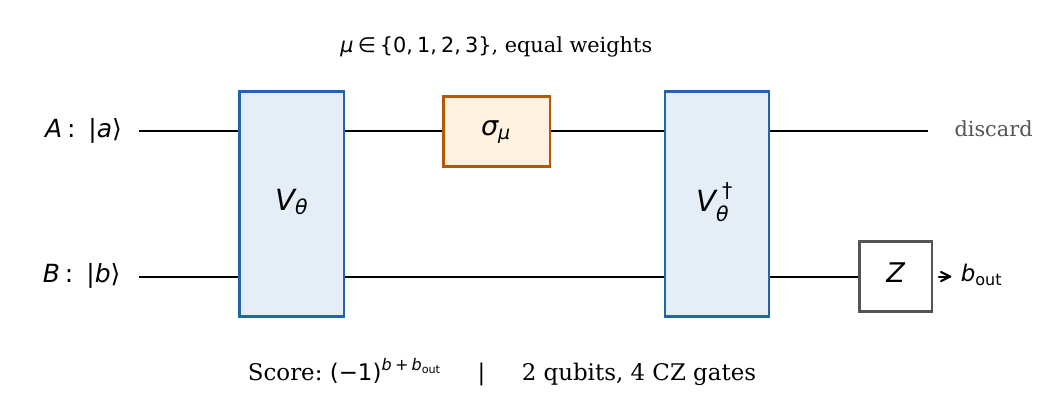}
 \caption{Reset-echo protocol. The four computational-basis inputs
 and four Pauli masks are equally weighted. The final score is
 $(-1)^{b_{\rm in}+b_{\rm out}}$. The protected setting uses
 $\theta=\pi/6$; the control uses $\theta=\pi/4$. No-reset references
 replace the mask by the identity while retaining the compiled
 forward/inverse gate pattern. Qubit $A$ is not measured.}
 \label{c-fig:circuit}
\end{figure}

Figure~\ref{c-fig:circuit} shows the protocol. The protected setting is the conserving gate at $\theta=\pi/6$. The control uses the same parameterized template at $\theta=\pi/4$,
keeping the target $H$ and $F$ of the $\pi/6$ setting fixed, so by the same computation its normalized correlation is $(c^2-s^2)^2=\cos^2(2\theta)$. Its
ideal normalized correlation is $\cos^2(2\theta)=0$, whereas the
protected setting gives $1/4$. The control gate does not conserve
the target $H$, although it conserves its own angle-dependent
observable. Thus the comparison is between two specified dynamics,
not a claim that the control has no conservation laws. Both
settings also have a no-reset echo whose ideal correlation is one.
The numerical values of the angle and the mask are substituted into the compiled circuit at run time, so the compiler cannot cancel the inverse pair in those reference circuits.

\subsection{Estimator and statistical analysis}
Prepare the four basis states $|ab\rangle$, $a,b\in\{0,1\}$, with equal weights and
measure $\sigma_z$ on qubit $B$ at the output. Here $b_{\rm in}=b$ is the input bit of $B$ (the bit $a$ of $A$ is averaged over) and $b_{\rm out}$ is the measured bit, whose $\sigma_z$ eigenvalue is $(-1)^{b_{\rm out}}$, so $\E[(-1)^{b_{\rm out}}\mid ab]=\Tr[Z\Phi(|ab\rangle\langle ab|)]$. For the score
$S=(-1)^{b_{\rm in}+b_{\rm out}}$, using $\sum_{a,b}(-1)^b|ab\rangle\langle ab|=Z$ and $F=cZ$, the ideal preparation/readout
identity is
\begin{align}
 \E S
 &=\frac14\sum_{a,b}(-1)^b
       \Tr[Z\Phi(|ab\rangle\langle ab|)]\nonumber\\
 &=\frac14\Tr[Z\Phi(Z)]
 =\tau[Z\Phi(Z)]=\frac{C_F}{w}.
 \label{c-eq:estimator}
\end{align}
This estimator identity holds for any fixed implemented channel,
even if it is nonunital (that is, $\Phi(I)\ne I$), because it is linear in $\Phi$ and uses only the inputs and the measured observable, provided the preparations and measurements
are ideal. That fact does not extend the theorem to nonunital
channels; actual preparation and readout errors remain possible.

The design was fixed before submission. A setting is one of $\{\text{protected},\text{control}\}\times\{\text{reset},\text{no reset}\}$, a row is one (setting, input state, mask) circuit with 128 shots, and a block is one copy of the 64-row grid; there were eight blocks, each containing
all four settings, all four initial basis states, and all four mask
labels. In no-reset settings the mask label
is ignored. There are 512 rows and 65,536 shots in total, or
$N_s=16{,}384$ per setting. The grid is balanced within each block;
its requested order was shuffled. No independence of the two
endpoint bits is assumed.

Let $\widehat C$ be the mean of $S$ over the $N_s$ shots of a setting. For independent shots conditional on the fixed contexts, Hoeffding's
inequality for scores in $[-1,1]$ (range $2$, hence the exponent $N_sh^2/2$) gives \cite{Hoeffding}
\begin{equation}
 \Prob\{|\widehat C-\E\widehat C|\ge h\}
 \le2e^{-N_sh^2/2}.
\end{equation}
A union bound over four settings sets
\begin{equation}
 8e^{-N_sh^2/2}=0.05,\qquad
 h=\sqrt{\frac{2\log160}{16{,}384}}=0.02489032048.
 \label{c-eq:CI}
\end{equation}
These are simultaneous 95\% sampling intervals, allowing different
means for the different input contexts. They do not certify absence
of correlated device drift or systematic errors. The fixed checks
were: a protected-reset lower endpoint above $0.20$, separation of
the protected and control reset intervals, and both no-reset lower
endpoints above the engineering threshold $0.90$.

\subsection{Measured correlations}
The run completed on September 18, 2026, on \texttt{ibm\_kingston},
physical qubits 49 and 50, under job
\texttt{damilsf8gn2s739lg7j0}, with 20 seconds of charged QPU time.
Table~\ref{c-tab:results} reports all four
settings; Figure~\ref{c-fig:data} compares them with the ideal predictions.
All 65,536 shots are included, without postselection, error mitigation,
fitted visibility, or division by a reference correlation.

This job followed several exploratory IBM runs on related boundary-memory models, none of which demonstrated the spatial mechanism; they are disclosed in Appendix~\ref{c-app:earlier} and are not pooled with Table~\ref{c-tab:results}.

\begin{table}[tb]
\centering\small
\begin{tabular}{@{}lrrrr@{}}
\toprule
Setting & Same sign & Opposite sign & Correlation & Simultaneous 95\% interval\\
\midrule
Protected, reset & 10,299 & 6,085 & $0.2572$ & $[0.2323,0.2821]$ \\
Control, reset & 8,164 & 8,220 & $-0.0034$ & $[-0.0283,0.0215]$ \\
Protected, no reset & 16,061 & 323 & $0.9606$ & $[0.9357,0.9855]$ \\
Control, no reset & 16,028 & 356 & $0.9565$ & $[0.9317,0.9814]$ \\
\bottomrule
\end{tabular}

\caption{Raw sign counts ($S=+1$ and $S=-1$ shots) and normalized correlations. Every setting has
16,384 shots. The ideal predictions are $0.25$, $0$, $1$, and $1$,
respectively. ``Reset'' is the ensemble randomization of
equation~\eqref{c-eq:twirl}. The intervals use~\eqref{c-eq:CI}.}
\label{c-tab:results}
\end{table}

The protected-reset lower endpoint, $0.2323$, exceeds the ideal
theorem benchmark $0.20$. The protected and control reset means
differ by $0.2606$; their simultaneous intervals imply a difference
interval $[0.2108,0.3104]$, the difference of means $\pm2h$. Both means are compatible with their
ideal predictions. Both no-reset lower endpoints exceed $0.90$,
but their intervals exclude the ideal value one. The observed
readout-bit change rates in these echoes, the opposite-sign counts of the two no-reset settings (shots with $b_{\rm out}\ne b_{\rm in}$ although the ideal echo returns the input), are $323/16{,}384=1.97\%$
and $356/16{,}384=2.17\%$. These are errors of the whole echo
experiment, not individual gate infidelities.

\begin{figure}[tb]
 \centering
 \includegraphics[width=\linewidth]{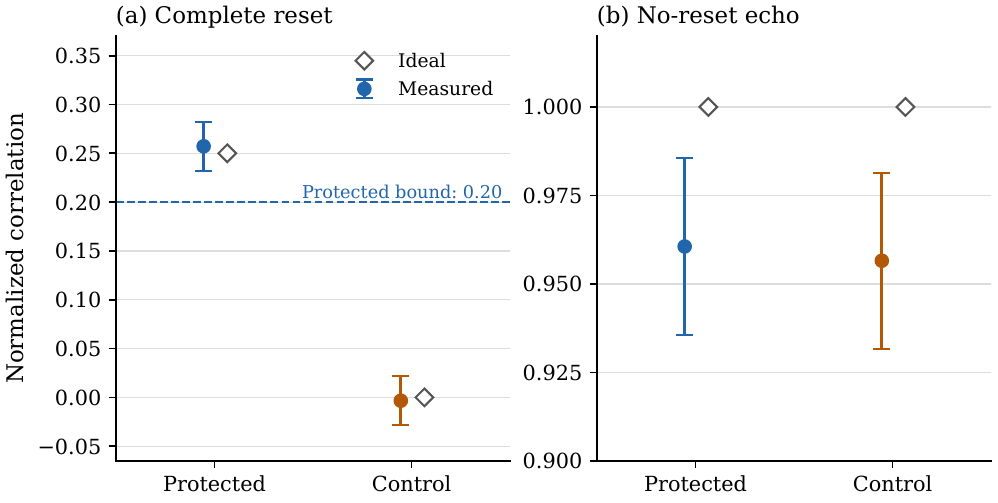}
 \caption{IBM measurements and ideal predictions. Filled circles
 show observed means with the predeclared simultaneous 95\% intervals;
 open diamonds show ideal values. The dashed $0.20$ line in the
 reset panel is the theorem's benchmark for the protected setting
 only. The no-reset panel uses a separate vertical scale to show
 implementation error. No fit or visibility correction is applied.}
 \label{c-fig:data}
\end{figure}

This is a resolved control contrast consistent with the chosen
two-qubit model. It does not certify exact charge conservation or
the projection-channel hypothesis for the noisy device. Hardware
errors can raise or lower a correlation, so a value above the ideal
bound is an instance-level consistency observation; a value below
it would not refute the theorem without certified hypothesis and
error bounds. The slight excess of the point estimate over $0.25$
is within the reported interval. The experiment neither witnesses
entanglement nor measures arbitrary-state quantum-memory fidelity, and it is an instance of the general inequality only: it does not realize the matrix-product chain and tests neither its localization nor its length-dependent lifetime.

\section{Interpretation and limitations}\label{sec:discussion}
The memory mechanism is a separation of scales: a reset with buffer $r$ leaks at most $C_{\rm loc}^2e^{-2\kappa r}$ of the charge, exponentially small in the buffer, while the signal $v_*$ is of order one and independent of $L$. A remote perturbation
changes the global conserved charge only weakly in the reference ensemble (the uniform word law of Section~\ref{sec:model}),
while a suffix observable has a nonzero overlap with that charge. The reset
inequality translates these two properties into an observable signal for
each allowed circuit. It is therefore a useful design and analysis bound
for constrained open-system dynamics: it quantifies how far noise must be
from the readout~\eqref{eq:bufferwindow}, how often resets may occur~\eqref{c-eq:lifetime}, and how much imperfect
conservation can be tolerated~\eqref{eq:continuous}. It does not establish fast computation,
fault tolerance, or a practical implementation of the matrix-product chain.

There is a sharp distinction between retaining an observable and storing a
qubit. For example, complete qubit dephasing
$\mathcal D_Z(\rho)=(\rho+\sigma_z\rho\sigma_z)/2$ retains the $\sigma_z$ autocorrelation
exactly: $\mathcal D_Z(\sigma_z)=\sigma_z$, so $\ip{\sigma_z}{\mathcal D_Z(\sigma_z)}=1$. Yet applying it to one half
of $|\Phi^+\rangle=(|00\rangle+|11\rangle)/\sqrt2$ gives
\begin{equation}
 \rho_{\rm out}:=(\id\otimes\mathcal D_Z)(|\Phi^+\rangle\langle\Phi^+|)
 =\tfrac12(|00\rangle\langle00|+|11\rangle\langle11|),
 \qquad
 \langle\Phi^+|\rho_{\rm out}|\Phi^+\rangle=\tfrac12.
\end{equation}
This output is separable (the off-diagonal terms of $|\Phi^+\rangle\langle\Phi^+|$ change sign under $\sigma_z$ on one qubit and average to zero). A channel on the affected qubit alone cannot
restore its entanglement with the untouched first qubit, which here plays the role of a reference system, because a local channel maps separable states to separable states. Thus even perfect retention
of one observable does not imply preservation of unknown superpositions.
The number $1/2$ above is the entanglement fidelity of $\mathcal D_Z$ for the maximally mixed input, expressed as an overlap with $|\Phi^+\rangle$. Recovering the qubit with high entanglement fidelity would require a recovery channel that brings this overlap close to one. Recovery guarantees concern such quantities; relevant two-sided and near-optimal
recovery bounds are developed by Tyson~\cite[arXiv v4, Sec.~5, Theorem~44]{Tyson} and
Barnum--Knill~\cite{BarnumKnill}. Those results supply context for this distinction,
not a derivation of our reset inequality.
For an encoded quantum state, approximate error correction can also be
formulated through worst-case entanglement fidelity and a dual problem
on the environment~\cite{BenyOreshkov}. These criteria are stronger than
the single-observable correlation examined here.

The spatial theorem uses a uniform-word mean-square norm. For suitable words with $G_L(w)=I$ and anisotropic suffix products, the projector $Z(G_L(w))$ differs from $Z(B_r(w))$ by order one; only the uniform-word average is small (end of Section~\ref{sec:localization}). The generic constants depend on
classical random-product estimates, and the explicit thirteen-state
constants remain conservative. The eight-site certificate removes the
large-readout requirement for its readout $F=PH_L^{\rm G}$ of~\eqref{eq:conditionalgram}, the conditional mean of the charge given the last seven labels, but does
not provide a low-depth physical implementation. Finally, the measured correlations of the two reset settings ($\theta=\pi/6$ and $\pi/4$)
are compatible with the exact two-qubit predictions
at the reported sampling precision; the two no-reset echo settings resolve implementation
errors. Earlier exploratory implementations, including a seven-qubit
pilot that failed its predeclared criterion and a cancelled run, are
disclosed in Appendix~\ref{c-app:earlier}; they are not pooled with the
two-qubit data. The sampling intervals do not bound
preparation drift or the diamond distance of the implemented channel.

Theorem~\ref{thm:global} concerns the averaged classical kernel with a bath at one end and open boundaries; it does not address baths at both ends, periodic boundaries, or the filling-length questions of Wang et al.~\cite{Wang}, and the local theorems assert nothing for typical pure initial states under an individual circuit. Generic nonunital dissipation and interacting finite-temperature reference states are outside the present theorems (Appendix~\ref{app:ensembles}).

\paragraph{Related boundary observables.}
Long-lived local information need not arise from matrix-product constraints.
For example, P{\l}odzie\'n and Chwede\'nczuk~\cite{Plodzien}
derive a boundary plateau of local quantum Fisher information in an open
Kitaev chain from Majorana-mode localization. That is a different model
and information measure. Comparing such mechanisms requires specifying
both the allowed noise and the observable being preserved.

\section{Open questions}\label{sec:open}
The following questions separate stronger mathematical bounds from the
additional ingredients needed for a physical memory.

\paragraph{How close are the explicit spatial and temporal scales to optimal?}
For the fixed-image charge of the thirteen-letter model, define
\begin{equation}
 e_r=\sup_{L\ge r}\|\widetilde H_L-\widetilde F_r\|_2^2,\qquad
 \gamma_* =\liminf_{r\to\infty}\bigl(-r^{-1}\log e_r\bigr).
\end{equation}
Here $\gamma_*$ is the true localization rate of this charge. Theorem~\ref{thm:sharper} gives $e_r\le(e^{-r/1400})^2=e^{-r/700}$ uniformly in $L$, hence $\gamma_*\ge1/700$.
Determining $\gamma_*$, or obtaining substantially stronger certified
bounds on it, would directly reduce the sufficient buffer sizes. A second
question is whether a specified, nontrivially mixing product-conserving
local circuit has a relaxation time (the first $t$ at which $C_\ell(t)$ falls below half its initial value) of order $\delta_r^{-2}=C_{\rm loc}^{-2}e^{2\kappa r}$, the length of the certified window in~\eqref{eq:bufferwindow}. Corollary~\ref{cor:prefix} is a lower bound uniform
over conserving circuits; it cannot supply an upper bound for that entire
class, which includes identity dynamics and exactly isolated readouts.

\paragraph{What is the optimal physical reset/readout bound?}
For fixed overlap fraction $q$ and relative leakage $x$ as in~\eqref{c-eq:qx}, is
$2q/(1+x)-1$ the infimum of $C_F/\|F\|_2^2$ over finite-dimensional
channel sequences satisfying the hypotheses of Theorem~\ref{thm:overlapreset}, or do complete
positivity and projection-channel structure imply a stronger bound in
some parameter regimes? The proof optimizes over the ellipsoid~\eqref{c-eq:ellipsoid} containing
all admissible final operators. It does not show that every boundary point
of that ellipsoid can be reached by such a channel sequence. The two-qubit
family has an explicitly positive gap from the bound, so it alone does
not establish optimality. Resolving this question would identify which
improvements require new model information.

\paragraph{Can finite-chain certificates be implemented with short local circuits?}
The eight-site result specifies a charge, conditional readout, and allowed
channels. It leaves open a scalable compilation with local gates that both
mix nontrivially within typical charge sectors and realize the prescribed
readout at useful depth. An implementation should demonstrate correlation
retention while a matched constraint-breaking circuit loses the same
signal. Finding such a construction at sizes where the exact certificate
is positive would connect the mathematical example to an experimental
test of spatially protected memory, beyond the two-qubit inequality test.

\paragraph{Can nonunital local baths admit a comparable spatial guarantee?}
Appendix~\ref{app:ensembles} treats a common invariant product reference
state. A further question is whether a local, nonunital bath with a
correlated Gibbs stationary state admits both a suitable weighted
norm-loss estimate and an exponentially localized charge approximation.
Weak-coupling master equations provide a standard class of thermal-bath
models~\cite{Davies1974}. For a detailed-balance convention in which the
dissipative Heisenberg generator $\mathcal L$ is self-adjoint for
$\langle X,Y\rangle_\rho=\Tr(\rho X^\dagger Y)$, stationarity and
contractivity give $\mathcal D=-\mathcal L\ge0$ and
$\frac{d}{dt}\|X_t\|_{2,\rho}^2=-2\langle X_t,\mathcal D X_t\rangle_\rho$.
This supplies the continuous norm-loss identity used in
Theorem~\ref{thm:continuous}, but not the spatial approximation or locality
of the projection onto fixed observables when $\rho$ is correlated. Any extension must also account for the
depolarization obstruction of Proposition~\ref{prop:upper}, or an analogue of it, when noise directly damps the observed charge.

\paragraph{Can the mechanism protect quantum information?}
All matrix-product charges constructed here are diagonal and commute. Protecting a
logical qubit would require control of noncommuting logical observables,
and an encoding and recovery statement that preserves entanglement with
a reference. Can a local constrained model support such an operator
algebra with comparably small boundary leakage and a nontrivial interior
dynamics? The diagonal construction and the present hardware data give
no such guarantee. A positive result would require an additional
structure beyond the classical signal proved here.

\paragraph{Acknowledgment.}
IBM Quantum services supplied the hardware access used in this work
\cite{IBMQuantum}. The analysis and conclusions are the author's and do
not represent IBM's position.
\appendix
\section{Global relaxation and the group-walk consequence}\label{app:global}
For completeness we state the associated global result. It concerns typical basis states under the averaged classical dynamics, whereas the local-memory theorems concern a specified mixed preparation under each individual circuit.
A doubly stochastic kernel is a matrix of transition probabilities with
every row and column summing to one; it preserves the uniform law.
For probability distributions on words,
$\|p-q\|_{\rm TV}=\frac12\sum_w|p(w)-q(w)|
=\sup_A|p(A)-q(A)|$. Thus one event that still remembers the initial
condition supplies a lower bound on the distance from equilibrium.
For a kernel $P$, $P^t(w,\cdot)$ denotes the distribution after $t$ steps started from the word $w$.

\begin{theorem}[Typical basis states in averaged dynamics]\label{thm:global}
Assume $\A$ and $\A^T$ generate strongly irreducible proximal semigroups,
as in Section~\ref{sec:localization}. A product-preserving kernel moves
probability only between words with equal $G_L$; a prefix update changes
only the first $k$ letters, where $k$ is fixed independently of $L$.
Let $P_L$ be any doubly stochastic product-preserving kernel followed by
any doubly stochastic prefix update. There are $c,C>0,L_0$, depending
only on $\A$ and $k$, uniform over these kernels, and sets $W_L$ with
$\pi_L(W_L)\ge1-Ce^{-cL}$ such that, for every $L\ge L_0$,
\begin{equation}\label{eq:global}
 \norm{P_L^t(w,\cdot)-\pi_L}_{\rm TV}\ge\tfrac34-e^{-cL}
 \quad(w\in W_L,\ 0\le t\le\lfloor e^{cL}\rfloor).
\end{equation}
The good set may depend on the kernel.
\end{theorem}
\begin{proof}
Write $G=s_1uv^T+E$ with $Ev=0$ and $\|E\|=s_2(G)$, as in the proof of Theorem~\ref{thm:local}, so that $(I-vv^T)G^T=E^T$. For any invertible $J$, a selected top right singular vector $v'$ of
$JG$ and the corresponding left vector $u'$, $JGv'=s_1(JG)u'$; projecting $v'=G^TJ^Tu'/s_1(JG)$ onto $v^\perp$, using $s_1(JG)\ge s_1(G)/\|J^{-1}\|$ and $\|\mathcal Z(x)-\mathcal Z(y)\|_F=\sqrt2\,d([x],[y])$ (proof of Theorem~\ref{thm:local}), gives
\begin{equation}\label{eq:stability}
 \norm{Z(JG)-Z(G)}_F\le\sqrt2\operatorname{cond}(J)\sigma(G).
\end{equation}
A prefix replacement multiplies $G$ on the left by one of finitely many
$J=p'p^{-1}$, where $p,p'$ are products of $k$ letters. Define
\begin{equation}
 K_{\rm pre}=\max_{p,p'}\operatorname{cond}(p'p^{-1}),
 \qquad \operatorname{cond}(J)=\|J\|\,\|J^{-1}\|.
\end{equation}
This maximum is finite: there are at most $D^{2k}$ pairs of invertible
prefix products. By~\eqref{eq:gap}, every such replacement moves $Z$ by
at most $\sqrt2K_{\rm pre}e^{-a_0L}$ outside a set of mass $C_1e^{-b_0L}$.

The limiting law $\eta$ of $Z(G_L)$ is nonatomic (proof of Theorem~\ref{thm:local}) and supported on the compact set of rank-one projectors; these imply that a sufficiently
fine fixed grid in $\mathbb R^{m^2}$ has limiting mass below $1/8$
in small neighborhoods covering each cell. A finite covering and the
closed-set inequality for weak convergence ensure cell masses at most
$1/4$ for all sufficiently large $L$, uniformly in grid translation.
For a random translation of mesh $h$, a strip of width $s_L$ along
coordinate faces has expected mass at most $2m^2s_L/h$.
Choose a translation attaining this bound with $s_L=\sqrt2K_{\rm pre}e^{-a_0L}$.
Outside the strip and the exceptional set, every prefix replacement
keeps the same cell label. Bulk updates preserve the label exactly.
The word law before the prefix update is $\pi_L$, because the bulk kernel is doubly stochastic; stationarity therefore bounds the one-step label-change probability by
$C_0e^{-c_0L}$, $c_0=\min(a_0,b_0)$, where $C_0$ depends only on $\A$, $k$ and the mesh $h$.

A union bound up to $T=\lfloor e^{c_0L/3}\rfloor$ makes the stationary
probability of any change at most $C_0e^{-2c_0L/3}$. Markov's inequality
over the initial word leaves a good set of mass
$1-C_0e^{-c_0L/3}$ on which the conditional change probability is at
most $e^{-c_0L/3}$. The initial cell has equilibrium mass at most $1/4$.
Testing total variation on it proves~\eqref{eq:global}, simultaneously
for all times up to $T$, with $c=c_0/3$ and $C=C_0$.
\end{proof}

Independent Haar averaging within equal-product classes of contiguous
blocks sends a basis state to the uniform mixture of words with the same block product, and so gives exactly the required classical kernels on diagonal states.
So does instantaneous uniform mixing in sectors contained in product
fibres. These averaged kernels map basis states to diagonal states, so for a basis state the quantum trace-norm distance is twice the
total variation in~\eqref{eq:global}. Wang et al.\ define the thermalization time as the largest, over initial states supported in one Krylov sector, of the first time at which the trace-norm distance to the maximally mixed state falls below $1/2$~\cite[Eqs.~(4)--(6)]{Wang}. Every basis state is supported in
one such sector, and hence is admissible in that definition. The bound gives
$t_{\rm th}\ge\lfloor e^{cL}\rfloor+1$ for all sufficiently large $L$.
This statement concerns the averaged channel; it is distinct from the
individual-circuit local-memory theorem.

For~\eqref{eq:alphabet}, the group generated by $I+2E_{12}$ and
$I+2E_{21}$ is free: on the projective line, including infinity,
nonzero powers of $x\mapsto x+2$ map $|x|<1$ into $|x|>1$, and
nonzero powers of $x\mapsto x/(2x+1)$ do the reverse. The ping-pong
lemma (two maps with disjoint attracting regions generate a free group) shows these two letters generate a free group of rank two, and a group containing one is nonamenable. Let $\mu$ be the uniform letter law on $\A_m$ and $\mu^{*L}$ its $L$-fold convolution, the law of a product of $L$ independent letters. Kesten's results~\cite{KestenWalk,KestenAmenability}
give spectral radius $\rho<1$ for the convolution operator of a symmetric lazy step law supported on a generating set of a nonamenable group, so
every fibre has uniform mass at most
$\sup_g\mu^{*L}(g)\le\rho^L$. This is the standard group-sector
fragmentation mechanism also used in~\cite[Prop.~2]{Wang}.
The one-step letter law assigns mass $1/D$ to the identity letter;
this is the self-loop making the walk lazy, not the mass of an identity-product fibre.
For $m\ge3$, the transvections $I+E_{12}$ and $I+E_{13}$ commute, since $E_{12}E_{13}=0=E_{13}E_{12}$, and
generate $\mathbb Z^2$; hyperbolic groups contain no $\mathbb Z^2$ subgroup, so these groups are not word hyperbolic, and Wang et al.'s Theorem~4, which covers hyperbolic groups, does not apply to them.

On the left Cayley graph of $\SL(m,\mathbb Z)$ with alphabet~\eqref{eq:alphabet},
neighbors have the form $g\leftrightarrow ag$, $a\in\A_m$.
This graph simply records multiplication by one allowed letter. The
distribution $\mu^{*L}$ is often called the group's heat kernel at time $L$.
The two-sided vertex boundary $\partial_v A$ consists of endpoints
of edges with one endpoint in $A$ and the other outside $A$.
The same grid construction, using~\eqref{eq:stability} for these finitely
many left multipliers, yields sets $A_L$ of group elements (those whose projector $Z$ lies in the accumulated cells) with
\begin{equation}\label{eq:vertex}
 \tfrac14\le\mu^{*L}(A_L)<\tfrac12,\qquad
 \mu^{*L}(\partial_v A_L)\le C'e^{-c'L}.
\end{equation}
To obtain the mass interval, accumulate grid cells until their mass
first reaches $1/4$; each has mass at most $1/4$, so the total is
strictly below $1/2$. Every endpoint of a crossing edge is either
exceptional (has $\sigma(g)>e^{-a_0L}$) or in the strip, by applying the stability estimate with $J$ a letter or its inverse to that
endpoint. This directly bounds the two-sided vertex boundary; no
conversion from an edge-flow bound is used. The sets may depend on $L$.
Qualitative group-walk non-expansion is already proved by
Fr\k{a}czyk and van Limbeek~\cite[Thm.~2.4, published version]{FraczykVanLimbeek}.\footnote{The same statement is Theorem~2.3 in arXiv version~1.}
To our knowledge, Equation~\eqref{eq:vertex} is the first exponential rate for this class,
whereas their theorem applies qualitatively to all finitely generated groups.
It has the intended small-boundary
form of Wang et al.'s Conjecture~2; their printed Eq.~(40) has the
opposite inequality sign, unlike the surrounding discussion and
Theorem~4~\cite{Wang}.

\section{Other reference ensembles}\label{app:ensembles}
Uniform weights are not essential to the localization mechanism.
Let every letter have positive probability $p_a$, and replace $\pi_L$
by the product law $p^{\otimes L}(w)=\prod_ip_{a_i}$; for uniform $p$ this is $\pi_L$, $\rho_p$ below is $I/D^L$, and $\langle\cdot,\cdot\rangle_p$ is $\ip\cdot\cdot$. The hypotheses of Theorem~\ref{thm:local} are unchanged, since the support of the letter law is still $\A$; only the constants and the stationary measure change. The same random-product and conditioning arguments
give localization with constants depending on $p$. Define
$\rho_p=(\diag p)^{\otimes L}$ and
$\langle X,Y\rangle_p=\Tr(\rho_pX^\dagger Y)$. If a quantum channel $\Phi$ preserves
$\rho_p$, its Heisenberg adjoint $\Phi^*$, defined by $\Tr(\Phi(\rho)X)=\Tr(\rho\Phi^*(X))$, is unital and completely positive and contracts this norm: the Kadison--Schwarz inequality~\cite{Kadison,Choi} $\Phi^*(X)^\dagger\Phi^*(X)\le\Phi^*(X^\dagger X)$ (the operator Schwarz inequality of Section~\ref{c-sec:setting}) gives $\|\Phi^*X\|_{2,\rho_p}^2\le\Tr(\rho_p\Phi^*(X^\dagger X))=\Tr(\Phi(\rho_p)X^\dagger X)=\|X\|_{2,\rho_p}^2$, where $\|X\|_{2,\rho_p}^2=\langle X,X\rangle_p$. In this appendix define explicitly
\begin{equation}
 C_\ell=\Tr\big(\rho_p F_\ell\Phi^*(F_\ell)\big),\qquad
 e_j=\norm{(\mathcal E_j^*-\id)H_L}_{2,\rho_p}.
\end{equation}
Telescoping in the Heisenberg order, $\Phi^*X-X=\sum_j\mathcal E_1^*\cdots\mathcal E_{j-1}^*(\mathcal E_j^*-\id)X$, and using that each adjoint factor contracts, gives the finite-block
certificate~\eqref{eq:finitecert} in the weighted norm, with $H_L,F_r$ centered under $p^{\otimes L}$, $v_\ell=\|F_\ell\|_{2,\rho_p}^2$ and $d_{L,\ell}=\|H_L-F_\ell\|_{2,\rho_p}$. (The Schr\"odinger-order telescoping of Theorem~\ref{thm:linear} does not transfer, because $\Phi$ itself need not contract the weighted norm.) For $0\le\zeta\le1$
the state $\rho_p(I+\zeta F_r)$ is positive, since $|f_r|\le1$, and normalized, since $F_r$ is $p$-centered; its signal is $\Tr(F_r\Phi(\rho_p(I+\zeta F_r)))=\zeta C_r$.

This extension requires every channel to preserve the same reference
state. An arbitrary product-preserving unitary need not preserve
$\rho_p$ when $p$ is nonuniform. Theorem~\ref{thm:overlapreset} used that each reset is an orthogonal projection for $\ip\cdot\cdot$; in the weighted norm the same proof needs $P_j=P_j^*$ for $\langle\cdot,\cdot\rangle_p$, which the uniform reset~\eqref{c-eq:reset} does not satisfy for nonuniform $p$ (it does not even preserve $\rho_p$), whereas replacement of a block by its marginal $\rho_{p,A}$ has an orthogonal Heisenberg projection
$X\mapsto I_A\otimes\Tr_A[(\rho_{p,A}\otimes I)X]$. No
generic interacting finite-temperature or amplitude-damping robustness
theorem follows merely by changing the norm.

\section{Explicit constants for the thirteen-state model}\label{sec:explicit}
For $\A_3$ a different charge, the fixed-image charge of Table~\ref{tab:charges}, admits fully numerical localization
constants. Send the fixed input $e_1$ through $G^T$, so that $G^Te_1$ is the first row of $G$, and define
\begin{equation}\label{eq:rationalcharge}
 u(G)=\frac{[(G^Te_1)_1]^2}{\norm{G^Te_1}^2}
      =\frac{G_{11}^2}{G_{11}^2+G_{12}^2+G_{13}^2}.
\end{equation}
Let $\widetilde H_L=\diag(u(G_L)-\E u(G_L))$ and $\widetilde F_r=\diag(u(B_r)-\E u(B_r))$, the latter extended by the identity outside the last $r$ sites, be the centered diagonal
operators obtained from $u(G_L)$ and $u(B_r)$, respectively.
They are rational-valued for integer labels, and since $u\in[0,1]$ both have $2$-norm at most $1/2$. Below $v_*$ denotes $\lim_{L\to\infty}\|\widetilde H_L\|_2^2$, which the proof identifies with the stationary coordinate variance of Proposition~\ref{prop:variance}. They differ at finite length
from the singular-direction charges, but have the same stationary
projective limit: $[G_L^Te_1]=[B_r^T(A^Te_1)]$ converges to the random line $\Lambda$ of~\eqref{eq:common} just as the top singular direction does. Conservation of $G_L$ preserves $\widetilde H_L$.
Normalized-Gram conservation alone need not preserve this alternative
charge; its conservation must be checked separately.

\begin{theorem}[Numerical localization constants]\label{thm:explicit}
For the uniform thirteen-letter alphabet and all $L\ge r\ge1$,
\begin{equation}\label{eq:explicitlocal}
 \norm{\widetilde H_L-\widetilde F_r}_2\le e^{-r/10000}.
\end{equation}
The limiting variance satisfies $v_*\ge1/756$, and
$|\norm{\widetilde F_r}_2-\sqrt{v_*}|\le e^{-r/10000}$.
Thus this charge has the explicit constants $C_{\rm loc}=1$ and $\kappa=10^{-4}$.
\end{theorem}

\begin{proof}
Let $g$ be one uniform letter and set $f(x)=\E\log\norm{gx}$ for
unit $x\in\mathbb R^3$. With $q_i=x_i^2$ and $S=q_1+q_2+q_3$ (kept as a symbol so that $P$ below is homogeneous, then evaluated at $S=1$): the letter $I\pm E_{ij}$ changes only coordinate $i$, $x_i\mapsto x_i\pm x_j$, so $\|(I\pm E_{ij})x\|^2=S+q_j\pm2x_ix_j$, and the two signs multiply to $(S+q_j)^2-4q_iq_j$; the identity contributes $\log1=0$, and $\log\|gx\|=\tfrac12\log\|gx\|^2$ over thirteen letters gives the factor $1/26$. Thus
pairing the two signs of each transvection gives
\begin{equation}\label{eq:stretchpoly}
 f(x)=\frac1{26}\log P(q),\qquad
 P(q)=\prod_{i\ne j}\big((S+q_j)^2-4q_iq_j\big),\quad S=1.
\end{equation}
The homogeneous degree-12 polynomial has the exact bounds
\begin{equation}\label{eq:polybounds}
 \frac{26}{5}\le P(q)\le16\qquad(q_i\ge0,\ S=1).
\end{equation}
For clarity, a complete coefficient certificate is provided below.

For distinct projective inputs represented by unit vectors $x,y$, let
$w=(x\times y)/\norm{x\times y}$ and
$R_g=d([gx],[gy])/d([x],[y])$. For unit vectors $d([x],[y])=\|x\times y\|$, and $gx\times gy=(\det g)\,g^{-T}(x\times y)$ with $g^{-T}=(g^{-1})^T$; since $\det g=1$,
\begin{equation}
 R_g=\frac{\norm{g^{-T}w}}{\norm{gx}\norm{gy}}.
\end{equation}
The letter law is invariant under $g\mapsto g^{-T}$, since $(I\pm E_{ij})^{-T}=I\mp E_{ji}$ is again a letter, hence
\begin{align}
 \E\log R_g
 &=f(w)-f(x)-f(y)\nonumber\\
 &\le\frac{\log16-2\log(26/5)}{26}
 =-\frac{\log(13/10)}{13}
 \le-\frac6{299}<-\frac1{50}.\label{eq:logdrift}
\end{align}
We used $\log z\ge2(z-1)/(z+1)$ for $z\ge1$.

The projective sine-distance ratio lies between
$\operatorname{cond}(g)^{-1}$ and $\operatorname{cond}(g)$, where $\operatorname{cond}(g)=s_1(g)/s_3(g)$ is the ratio of largest to smallest singular value.
Indeed, restrict $g$ to the plane spanned by $x,y$; its area expansion
is the product of its two singular values and each vector's norm is
at least the smaller one. Apply the same argument to $g^{-1}$ for
the lower bound. For an elementary shear the condition number is
$(3+\sqrt5)/2<e$, while the identity has condition number one.
Therefore $X=\log R_g\in[-1,1]$.

For a random variable in $[-1,1]$, its log moment-generating function $\psi(s)=\log\E e^{sX}$
has second derivative equal to the variance of $X$ under the tilted law with density $e^{sX}/\E e^{sX}$, at most one; this is Hoeffding's lemma~\cite{Hoeffding}.
Consequently $\log\E e^{sX}\le s\E X+s^2/2$ for $s\ge0$.
Taking $s=1/50$ and using~\eqref{eq:logdrift} yields the uniform
one-step certificate
\begin{equation}\label{eq:explicitq}
 \sup_{[x]\ne[y]}\E
 \left(\frac{d([gx],[gy])}{d([x],[y])}\right)^{1/50}
 \le e^{-1/5000}.
\end{equation}
Indeed the logarithmic upper bound is at most
$-301/1495000<-299/1495000=-1/5000$.

Write $G_L=AB_r$, so that $G_L^Te_1=B_r^T(A^Te_1)$ and $B_r^Te_1$ are the images of the two inputs $A^Te_1$ and $e_1$ under the same product $B_r^T=a_L^T\cdots a_{L-r+1}^T$. Conditional on the prefix, this input pair
is independent of the transpose suffix. Transposition
preserves the letter law. Applying~\eqref{eq:explicitq} to one letter at a time, conditional on the letters already applied, and using the tower property $r$ times, with initial
distance at most one, gives
\begin{equation}
 \E d([G_L^Te_1],[B_r^Te_1])^{1/50}\le e^{-r/5000}.
\end{equation}
Coincident inputs remain coincident and need no ratio. For unit $v,z$,
$|v_1^2-z_1^2|\le\norm{vv^T-zz^T}=d([v],[z])$.
Since $d^2\le d^{1/50}$ on $[0,1]$, the raw charge difference has
second moment at most $e^{-r/5000}$. Centering contracts $L^2$, and taking square roots proves
\eqref{eq:explicitlocal}.

Finally, couple $e_1$ to an independent stationary projective input, a random line $\Lambda_0$ with the stationary law of the transpose walk (Section~\ref{sec:localization}), realized on the same probability space and independent of the letters,
and pass both through the same transpose product. By stationarity the second output
has stationary law at every length; the same estimate bounds their centered scalar
difference in $L^2$ by $e^{-r/10000}$, and the reverse triangle inequality gives the bound on $|\|\widetilde F_r\|_2-\sqrt{v_*}|$. The limiting variance is thus that of
the stationary projector coordinate, to which
Proposition~\ref{prop:variance} applies.
\end{proof}

\paragraph{Exact polynomial certificate.}
Write, with multi-index $\alpha=(a,b,c)$, $|\alpha|=a+b+c$, $q^\alpha=q_1^aq_2^bq_3^c$ and $\binom{12}{\alpha}=12!/(a!b!c!)$,
$P(q)=\sum_{|\alpha|=12}c_\alpha q^\alpha
      =\sum_{|\alpha|=12}b_\alpha\binom{12}{\alpha}q^\alpha$.
The Bernstein weights $\binom{12}{\alpha}q^\alpha$ are nonnegative
and sum to one on $S=1$ by the multinomial theorem, so $P(q)$ is a weighted average of the $b_\alpha$ and lies between their minimum and maximum (Appendix~\ref{app:reproducibility}). The integer $c_\alpha$ and the rational $b_\alpha$ are produced by \texttt{anc/certify\_one\_step.py}; the table lists $b_\alpha$. Permutation symmetry reduces the 91 coefficients
to the following 19 values; the remaining values are their permutations.
They follow by multiplying the six quadratics in~\eqref{eq:stretchpoly}
with integer arithmetic.
\begin{center}\small
\begin{tabular}{c r@{\qquad}c r}
$\alpha$ & $b_\alpha$ & $\alpha$ & $b_\alpha$\\\hline
$(12,0,0)$ & $16$ & $(7,5,0)$ & $86/11$\\
$(11,1,0)$ & $12$ & $(7,4,1)$ & $1036/165$\\
$(10,2,0)$ & $28/3$ & $(7,3,2)$ & $2813/495$\\
$(10,1,1)$ & $296/33$ & $(6,6,0)$ & $3659/462$\\
$(9,3,0)$ & $441/55$ & $(6,5,1)$ & $4306/693$\\
$(9,2,1)$ & $1211/165$ & $(6,4,2)$ & $5491/990$\\
$(8,4,0)$ & $3809/495$ & $(6,3,3)$ & $2092/385$\\
$(8,3,1)$ & $361/55$ & $(5,5,2)$ & $11593/2079$\\
$(8,2,2)$ & $3089/495$ & $(5,4,3)$ & $1657/315$\\
 & & $(4,4,4)$ & $10061/1925$\\\hline
\end{tabular}
\end{center}
The minimum is $10061/1925>26/5$ and the maximum is $16$,
proving~\eqref{eq:polybounds} without sampling the simplex.

\paragraph{Numerical implications and their scale.}
Lemma~\ref{lem:spatial}, Theorem~\ref{thm:linear}, Corollary~\ref{thm:reset} and Corollary~\ref{cor:prefix} apply to the tilded charges with $C_{\rm loc}=1$ and $\kappa=10^{-4}$, since they use only $\|\widetilde H_L-\widetilde F_r\|_2\le\delta_r$, the $2$-norm bound $1/2$, and the suffix support of $\widetilde F_r$.
For $v_{\min}=1/756$, \eqref{eq:ell0} gives the explicit sufficient readout length
\begin{equation}
 \ell_0=\lceil10000\log6048\rceil=87075.
\end{equation}
If $L\ge\ell\ge\ell_0$, the bistochastic channels between resets preserve
$\widetilde H_L$, one partial prefix reset $(1-p)\id+pP$ of strength $p\in(0,1]$ per round
leaves a buffer of $r$ sites with $1\le r<L$, the readout is $\widetilde F_\ell$ and $t$ counts rounds, then
\begin{equation}\label{eq:explicitlife}
 C_\ell(t)\ge\frac1{1512}
 \quad\text{for integer}\quad
 0\le t\le\frac{2-p}{12096p}e^{r/5000}.
\end{equation}
For complete resets and $r=87075$, since $87075>10000\log6048$ gives $e^{r/5000}>6048^2$ and $6048^2/12096=3024$, this certifies at least 3024
rounds, with $L\ge87076$ if the reset region is nonempty.
These are conservative sufficient values. They remove the unspecified
constants for this observable but do not provide a practical small-device
prediction. Sharpening the logarithmic drift, its fluctuation bound and
the finite-block readout remains necessary for that purpose.
Appendix~\ref{sec:sharper} sharpens all three parts of this estimate.

\section{Sharper contraction, signal, and readout constants}\label{sec:sharper}
The same rational charge admits improved constants by retaining the
finite-alphabet structure and using the readout overlap directly.
All statements in this section concern the uniform thirteen-letter law.

Here $v_*$ is the limiting variance of Appendix~\ref{sec:explicit}, the same stationary-coordinate variance as in Theorem~\ref{thm:local} for coordinate $j=1$.
\begin{theorem}[Improved localization and signal]\label{thm:sharper}
For the charge~\eqref{eq:rationalcharge} and all $L\ge r\ge1$,
\begin{equation}\label{eq:sharperloc}
 \norm{\widetilde H_L-\widetilde F_r}_2\le e^{-r/1400},
 \qquad v_*\ge\nu:=\frac1{19}.
\end{equation}
Writing $v_n=\norm{\widetilde F_n}_2^2$, one also has
\begin{equation}\label{eq:sharpervariance}
 |v_n-v_*|\le2e^{-n/700}.
\end{equation}
\end{theorem}

The audit program also records the weaker intermediate choice $s=4/125$, $\kappa=1/4800$; the improved certificate below establishes the rate $1/1400$ used in the main results.

\subsection{Contraction using orthogonality}
The polynomial $P$ in~\eqref{eq:stretchpoly}, which encodes the one-letter stretch $f(x)=\E\log\|gx\|=\tfrac1{26}\log P(q)$ with $q_i=x_i^2$, satisfies the sharper bound
\begin{equation}\label{eq:sharpP}
 \frac{4096}{729}\le P(q)\le16\qquad(S=1).
\end{equation}
The lower value is attained when $q_1=q_2=q_3=1/3$.
Here is an exact certificate: by permutation symmetry it suffices to put
\begin{equation}\label{eq:cone}
 q_1=x+y+z,\qquad q_2=y+z,\qquad q_3=z,
 \qquad x,y,z\ge0.
\end{equation}
Because $P$ is symmetric under permutations of $(q_1,q_2,q_3)$, it suffices to treat $q_1\ge q_2\ge q_3\ge0$, which~\eqref{eq:cone} parametrizes with $x=q_1-q_2$, $y=q_2-q_3$, $z=q_3$. The homogeneous polynomial $729P-4096S^{12}$ then has 88 nonzero
coefficients, all positive; a polynomial with nonnegative coefficients is nonnegative for $x,y,z\ge0$, and dividing by $S^{12}$ proves the lower bound on $S=1$. The only absent degree-12 monomials have
$(x,y)$ degree zero or one, consistent with the form vanishing at the minimizer $x=y=0$. Integer multiplication of the six displayed
quadratics verifies this certificate; all coefficients and an independent
verifier are included in the ancillary files. The upper bound is already
certified in Appendix~\ref{sec:explicit}.

The normal $w$ in the cross-product identity is perpendicular to both
input directions. Retaining this relation gives the stronger estimate
\begin{equation}\label{eq:orthogonalP}
 \frac{P(q)^2}{P(r)}\ge\frac72
 \quad\text{if }q,r\text{ lie in the simplex }\{q_i\ge0,\ \textstyle\sum_iq_i=1\}\text{ and }q_i+r_i\le1.
\end{equation}
Here is a finite algebraic proof. The degree-10 polynomial
\begin{equation}
 J(q)=\frac{\partial_1P(q)-\partial_2P(q)}{q_1-q_2}
\end{equation}
is a polynomial because $\partial_1P-\partial_2P$ is antisymmetric under $q_1\leftrightarrow q_2$ and hence divisible by $q_1-q_2$; it has 66 coefficients, all positive, with minimum 48. Symmetry gives
the same sign for every coordinate pair. Since $(q_1-q_2)(\partial_1P-\partial_2P)=(q_1-q_2)^2J\ge0$, $P$ is Schur-convex:
moving mass from a smaller coordinate to a larger one increases $P$,
and $P(q)\ge P(r)$ whenever $q$ majorizes $r$ (sorted partial sums of $q$ dominate those of $r$).
For three coordinates of equal sum, majorization means that, after
sorting, the largest coordinate of $q$ is at least that of $r$, and
the smallest coordinate of $q$ is at most that of $r$. The derivative
test above proves monotonicity under the pairwise transfers generating
this order, so no unproved property of $P$ is assumed.

Put $a=\max_i r_i$ (here $a$ is a scalar and $r$ a simplex point; neither denotes the overlap or buffer of the main text). If $a\le2/3$, then $r$ is majorized by
$(2/3,1/3,0)$, so $P(r)\le P(2/3,1/3,0)=54400/6561$ by Schur-convexity, and equation~\eqref{eq:sharpP} gives
\begin{equation}
 \frac{P(q)^2}{P(r)}\ge
 \frac{(4096/729)^2}{P(2/3,1/3,0)}
 =\frac{131072}{34425}>\frac72.
\end{equation}
If $a\ge2/3$, relabel so that $r_1=a$. Then $q_1\le1-a$,
so the largest coordinate of $q$ is at least $a/2$ and its smallest
is at most $1-a$. Consequently
\begin{equation}
 P(q)\ge P(1-a,a/2,a/2),\qquad
 P(r)\le P(a,1-a,0).
\end{equation}
The remaining interval inequality is certified homogeneously.
Set $a=(2u+3v)/(3(u+v))$, $u,v\ge0$, $u+v>0$, and
\begin{equation}
 Q=(2u,2u+3v,2u+3v),\quad R=(4u+6v,2u,0),\quad S=6(u+v).
\end{equation}
All 25 coefficients of the degree-24 binary form
$2P(Q)^2-7P(R)S^{12}$ are positive. Dividing by $S^{24}$ proves
\eqref{eq:orthogonalP}. Both coefficient lists and their exact verifier
are included in the ancillary files.

For orthogonal unit vectors $z,w$, the squared-coordinate vectors
satisfy $z_i^2+w_i^2\le1$, as diagonal entries of an orthogonal
rank-two projector. Thus, with $q_i=z_i^2$, $r_i=w_i^2$ and $f=\tfrac1{26}\log P$, \eqref{eq:orthogonalP} says
$2f(z)-f(w)\ge\log(7/2)/26$. Apply this to $z=x$ and $z=y$ and
average. For the logarithmic distortion $X_g=\log R_g$, this gives
\begin{equation}
 \E X_g=f(w)-f(x)-f(y)\le-\frac{\log(7/2)}{26}<-\frac6{125}.
\end{equation}
The last comparison follows from the first four terms of
$\log(7/2)=2\sum_{k\ge0}(5/9)^{2k+1}/(2k+1)$, the series for $\log\frac{1+z}{1-z}$ at $z=5/9$; all terms are positive, so a truncation is a lower bound.
Furthermore $|X_g|<c:=77/80$. Indeed $|X_g|\le\log\operatorname{cond}(g)$ (Appendix~\ref{sec:explicit}) and, for a shear, $\operatorname{cond}(g)=(3+\sqrt5)/2<(3+161/72)/2=377/144<\sum_{k=0}^{6}(77/80)^k/k!\le e^{c}$, using
$\sqrt5<161/72$.

The identity contributes $X_I=0$. The twelve other logarithmic
distortions lie in $[-c,c]$ and have sum at most $-\Delta$, where
$\Delta=78/125=13\cdot6/125$, since $13\,\E X_g=\sum_{g\ne I}X_g$. For $s>0$, convexity gives
\begin{equation}\label{eq:finiteletter}
 \E e^{sX_g}\le
 \frac{1+6e^{-cs}+5e^{cs}+e^{(c-\Delta)s}}{13}.
\end{equation}
To see this, first increase coordinates until their sum is $-\Delta$; this can only increase $\sum e^{sx_g}$ for $s>0$.
On the resulting slice of the cube $[-c,c]^{12}$ by the hyperplane of sum $-\Delta$, a convex function has a maximum at a vertex; a vertex has eleven of its twelve coordinates at $\pm c$, so at most one coordinate is interior. An all-at-bounds vertex
would have an even integer multiple of $c$ as its sum, which cannot
equal $-\Delta$. Since $0<\Delta<2c$,
that vertex has six coordinates $-c$, five coordinates $c$, and
one coordinate $c-\Delta$.

Take $s=3/50$. A rational bound for the right-hand side is
\begin{equation}
 \E e^{sX_g}\le
 \frac{379633260899556669642483040067}{380186475500000000000000000000}
 <\frac{699}{700}=1-\frac1{700}<e^{-1/700}.
\end{equation}
This comparison requires no transcendental numerical evaluation:
for $t\le0$ use $e^t\le\sum_{k=0}^4t^k/k!$, and for $0\le t<1$
use
\begin{equation}
 e^t\le\sum_{k=0}^4\frac{t^k}{k!}
          +\frac{t^5}{120(1-t/6)}.
\end{equation}
Substitution in~\eqref{eq:finiteletter} proves the stated rational bound.
The shared-suffix argument of Appendix~\ref{sec:explicit} therefore gives
both
\begin{equation}\label{eq:sharperraw}
 \E|u(G_L)-u(B_r)|^2\le e^{-r/700},\qquad
 \E|u(G_L)-u(B_r)|\le e^{-r/700}.
\end{equation}
The second estimate uses $d\le d^{3/50}$ on $[0,1]$;
the first uses $d^2\le d^{3/50}$. Centering and taking a square root
prove the localization part of~\eqref{eq:sharperloc}.

Coupling instead to the unique stationary law used in
Proposition~\ref{prop:variance}, with $U$ the squared first coordinate of $[B_n^Te_1]$ (so $\Var U=v_n$) and $V$ the same for $[B_n^T\xi]$, $\xi$ an independent line of stationary law (so $V$ is stationary and $\Var V=v_*$), gives $\E|U-V|\le e^{-n/700}$, the same first-moment
estimate for the finite and stationary coordinate variables.
For any coupled $U,V\in[0,1]$,
\begin{equation}
 |\Var U-\Var V|\le2\E|U-V|.
\end{equation}
Indeed use $\E V$ as a trial centering constant for $U$, bound the
difference of the two squares, and then interchange $U,V$.
This proves~\eqref{eq:sharpervariance}.

\subsection{A stronger stationary variance certificate}
For a projective vector let $q_i=x_i^2/\norm x^2$ and define
$U(q)=\sum_iq_i^2$. This is the concentration of its squared coordinates,
not the purity of its rank-one quantum projector. We certify the
pointwise bound
\begin{equation}\label{eq:outgoingU}
 \frac1{12}\sum_{i\ne j}\sum_{\epsilon=\pm1}
 U\big([(I+\epsilon E_{ij})x]\big)\ge\frac{28}{57}.
\end{equation}
For completeness, the following rational polynomial construction specifies
the entire finite certificate. Allow homogeneous $q_i\ge0$, put
$S=q_1+q_2+q_3$, and for $\{i,j,k\}=\{1,2,3\}$ define
\begin{align}
 A_{ij}&=S+q_j,& R_{ij}&=q_j+q_k,& H_{ij}&=4q_iq_j,\\
 B_{ij}&=A_{ij}^2-H_{ij},\\
 N_{ij}&=B_{ij}^2-2R_{ij}A_{ij}B_{ij}
 +(R_{ij}^2+q_j^2+q_k^2)(A_{ij}^2+H_{ij}).
\end{align}
The average of the two signed outgoing values of $U$ is
$N_{ij}/B_{ij}^2$. This follows by writing each outgoing value as
\begin{equation}
 1-\frac{2R_{ij}}{A_{ij}\pm2\sqrt{q_iq_j}}
 +\frac{R_{ij}^2+q_j^2+q_k^2}
       {(A_{ij}\pm2\sqrt{q_iq_j})^2}
\end{equation}
and averaging the signs. The denominators are positive for $S>0$.
Set, in this subsection only ($D$ and $T$ here are polynomials, not the site dimension or a channel),
\begin{equation}
 D(q)=\prod_{i\ne j}B_{ij}^2,\qquad
 T(q)=\sum_{i\ne j}N_{ij}\prod_{(a,b)\ne(i,j),\,a\ne b}B_{ab}^2.
\end{equation}
The expression on the left of~\eqref{eq:outgoingU} is $T/(6D)$.
After~\eqref{eq:cone}, the degree-24 polynomial $57T-168D$ has
325 coefficients, all strictly positive; the smallest is $29952$.
This proves~\eqref{eq:outgoingU} on the ordered cone, and permutation
symmetry proves it everywhere. The full coefficient list and exact
integer-arithmetic verifier are supplied as ancillary files.

At stationarity, $\E U=\tfrac1{13}\E U+\tfrac{12}{13}\E[T/(6D)]$ (the identity letter contributes $U$ itself), so
$\E U=\E(T/(6D))\ge28/57$; here the stationary law is that of the transpose walk, and $\A_3^T=\A_3$. Permutation symmetry of the unique stationary
law gives $\E q_i=1/3$ and $\E q_i^2=\E U/3$. Consequently
\begin{equation}
 v_* = \frac{\E U}{3}-\frac19\ge\frac1{19},
\end{equation}
completing the proof of Theorem~\ref{thm:sharper}.

\subsection{Improved finite readout and reset window}
\begin{corollary}\label{cor:sharperwindow}
Let
\begin{equation}
 \ell_{\rm sharp}=\lceil700\log190\rceil=3673.
\end{equation}
For $L\ge\ell\ge\ell_{\rm sharp}$, alternate arbitrary
$\widetilde H_L$-preserving bistochastic channels with one partial reset per round
of strength $p\in(0,1]$ on the first $L-r$ sites, where $1\le r<L$.
Then
\begin{equation}\label{eq:sharperwindow}
 C_\ell(t)\ge\frac1{38}\qquad
 \text{for integers }0\le t\le\frac{2-p}{304p}e^{r/700}.
\end{equation}
The readout is $\widetilde F_\ell$.
\end{corollary}
\begin{proof}
Put $\eta=e^{-\ell/700}\le\nu/10$, $v=\norm{\widetilde H_L}_2^2$,
$w=\norm{\widetilde F_\ell}_2^2$, and
$a=\ip{\widetilde H_L}{\widetilde F_\ell}$.
The preceding estimates give
$v,w\ge\nu-2\eta\ge4\nu/5$ (from~\eqref{eq:sharpervariance} at $n=L$ and $n=\ell$, with $e^{-L/700}\le\eta$) and
$a\ge(v+w-\eta)/2>0$ (from $\|\widetilde H_L-\widetilde F_\ell\|_2^2=v+w-2a\le\eta$); the margins are $190=10/\nu$ and $304=16/\nu$.
The reset budget of~\eqref{c-eq:parameters} satisfies $\B\le tp e^{-r/700}/(2-p)\le\nu/16$, since each of the $t$ resets has $b_j^2\le e^{-r/700}$ by Lemma~\ref{lem:spatial} applied to the tilded charge and~\eqref{eq:sharperraw}.
Theorem~\ref{thm:overlapreset} therefore implies
\begin{equation}
 C_\ell(t)\ge\frac{(v+w-\nu/10)^2}{2(v+\nu/16)}-w
 \ge\frac{58}{115}\nu>\frac\nu2.
\end{equation}
For the last inequality, the middle expression increases in $w$ for
$w\ge4\nu/5$; after setting $w=4\nu/5$ it increases in $v$ for
$v\ge4\nu/5$. Its value at $v=w=4\nu/5$ is $58\nu/115$.
\end{proof}

The readout length and the protected buffer play different roles.
For complete resets ($p=1$), \eqref{eq:sharperwindow} certifies $t$ rounds once $r\ge700\log(304t)$; the following sufficient buffers illustrate
the new window with $\ell=\ell_{\rm sharp}$ (the row 3024 is carried over from Appendix~\ref{sec:explicit} for comparison); a nonempty noisy prefix
requires $L\ge r+1$.
\begin{center}
\begin{tabular}{r r}
Guaranteed complete reset rounds & Sufficient buffer $r$\\\hline
1 & $\lceil700\log304\rceil=4002$\\
1000 & $\lceil700\log304000\rceil=8838$\\
3024 & $\lceil700\log919296\rceil=9612$\\\hline
\end{tabular}
\end{center}
The buffer $r=\ell_{\rm sharp}$ alone does not certify a complete reset
round. These conservative bounds reduce the previous sufficient sizes and increase
the guaranteed signal, but still do not imply a practical small-device
experiment. All ceilings are verified by rational logarithm enclosures.

The ancillary verifiers use only integer or rational polynomial arithmetic
and Taylor bounds. Their use is a finite computer-assisted proof, with no
direction sampling, fitted rates, or physical simulation.

\section{Earlier implementation studies}\label{c-app:earlier}
The two-qubit circuit followed deeper exploratory implementations of related
boundary-memory models. We report the earlier outcomes in Table~\ref{c-tab:earlier} to distinguish the
successful two-qubit consistency test from unsuccessful attempts to
demonstrate the spatial mechanism (the exponential localization of Theorem~\ref{thm:local}). The observables, controls, and confidence
rules differ between studies; their results are not pooled. Dates are UTC.
Resource counts refer to the principal memory circuit, not every calibration or witness circuit within the job; a witness circuit measures an auxiliary diagnostic observable. A cycle is one repetition of that study's prescribed conserving dynamics and boundary-reset sequence. The archived reports specify each preparation, readout score, and ordering of the reset and conserving dynamics; reported correlations are sample means of those scores.

Ideal targets quoted in the table come from each study's site model and archived analysis. Their closed forms are reconstructed at the end of this appendix, where the prespecified benchmark is distinguished from a stronger bound obtained after the design was fixed.

The pair-flip studies use four-state sites with involutive labels ($a^2=I$, so each label is its own inverse): a
conserving gate mixes equal pairs $|aa\rangle$ and $|bb\rangle$, both
of which have product $I$, and preserves unequal pairs. The readout is
a sign assigned to a letter of the reduced word, obtained by deleting
adjacent equal labels. The adapted study retains a five-site logical
chain but enumerates its input and bath basis states and uses a fresh
register for the bath instead of a mid-circuit reset. Its changed-selector
control alters which pair transitions are allowed and does not conserve
the original word. The paired-charge design compares the first and last
reduced-word signs in a three-site chain. The reduced matrix-product
pilot uses seven four-state sites labeled by matrices in $\SL(2,\mathbb Z)$,
with three active sites on the device and four sites averaged classically.

\begingroup\small
\begin{longtable}{@{}>{\raggedright\arraybackslash}p{0.27\linewidth}>{\raggedright\arraybackslash}p{0.25\linewidth}>{\raggedright\arraybackslash}p{0.40\linewidth}@{}}
\caption{Exploratory hardware record. Supporting reports, job identifiers and the cancellation record
are included in the reproduction bundle; raw records for these earlier
studies are not included.}\label{c-tab:earlier}\\
\toprule
Study and IBM job & Hardware resources & Result and scope\\
\midrule
\endfirsthead
\toprule
Study and IBM job & Hardware resources & Result and scope\\
\midrule
\endhead
\bottomrule
\endfoot
First pair-flip pilot, Sept.\ 16\newline
\texttt{dalf9tv8gn2s739jv45g} &
\texttt{ibm\_marrakesh}; 34 qubits; 984 CZ; depth 1,749; 7,168 total shots &
Two-cycle conserving correlation $0.1992$, but the one-cycle result
$0.1113$ was below its ideal bound $0.2515$ even with its sampling margin.
Control gate counts were unmatched. No protected-memory demonstration.\\[6pt]
Ten-qubit follow-up, Sept.\ 16\newline
\texttt{dalfh9f8gn2s739jvdrg} &
\texttt{ibm\_marrakesh}; 10 qubits; 225 CZ; depth 455; 49,152 total shots &
Raw conserving correlation $0.0020$ versus ideal lower bound $0.4223$;
matched identity $-0.0166$. Readout-only corrections were insufficient
under their calibration assumptions. No memory-protection demonstration.\\[6pt]
Adapted diagnostic, Sept.\ 16\newline
\texttt{dalfsd02fm4c73f1icug} &
\texttt{ibm\_kingston}; 16,384 shots; 8,192 per setting &
Preparation/readout reference $0.9602$ and full-depth identity $0.5991$
passed the prespecified engineering thresholds. This was a feasibility
check, not a memory-mechanism test.\\[6pt]
Adapted pair-flip test, Sept.\ 17\newline
\texttt{daljchdr85ps73fbhd5g} &
\texttt{ibm\_kingston}; 241 CZ in the memory template; 41,984 total shots &
Conserving $0.5415$, changed-selector $0.5210$, matched identity $0.5488$.
The conserving interval was $[0.5115,0.5715]$, above the planned
$0.4223$ target. Conserving-minus-selector and conserving-minus-identity
intervals were $[-0.0219,0.0629]$ and $[-0.0498,0.0351]$.
The stronger ideal bound $0.5841$ exceeded the conserving interval.\\[6pt]
Paired-charge attempt, Sept.\ 17\newline
\texttt{dalvpttr85ps73fc1lpg} &
\texttt{ibm\_kingston}; 12 charged QPU seconds &
Cancelled for exceeding the execution-time limit. No completed result
was analyzed; the failure message is archived.\\[6pt]
Reduced matrix-product pilot, Sept.\ 18\newline
\texttt{damadggpqrnc73976tsg} &
\texttt{ibm\_kingston}; 7 qubits; 222 CZ; depth 552; 16,384 shots &
Failed the predeclared primary criterion, a positive lower interval endpoint for the contrast (the reset gap minus the no-bath gap; ideal $0.00693$): contrast $0.00192$ with 95\%
interval $[-0.00080,0.00465]$. The active suffix (the three sites placed on the device) changed in $52.7\%$ of
no-reset echo shots, versus ideal zero.\\
\end{longtable}
\endgroup

The adapted memory job finished at 00:43:43 UTC on September 17
(20:43:43 EDT on September 16); its historical report uses the local date.
The seven-qubit pilot finished at 02:49:51 UTC on September 18
(22:49:51 EDT on September 17); its report likewise uses the local date.
In that pilot, the seven-site charge depends on all seven labels, but the four sites not placed on the device kept their prepared labels by construction and entered the score classically, so part of any protected correlation was never exposed to noise. It was not a full physical realization of the seven-site
model, and a high uncorrected protected score was not evidence of memory
protection. The two-qubit experiment has no such spectator (classically retained) contribution.
The readout-corrected exploratory estimates of the ten-qubit follow-up are retained in its
report, with their assumptions; they are not used to correct the
two-qubit data in Table~\ref{c-tab:results}.

The historical ideal benchmarks concern these pair-flip models, rather
than the thirteen-state matrix alphabet. For the 17-site pilot the
prespecified weaker bound was $7/18-2\sqrt{t\,p_{16}}$, where
$p_{16}=5068915/2^{30}$ is the 16-letter return probability; $t=1,2$
gives $0.251473$ and $0.194553$. For the five-site studies the formula was
\begin{equation}
 B_5(t)=\frac{25}{32}-\frac{2x}{1+x},\qquad x=\frac{7t}{64}.
\end{equation}
Thus $B_5(2)=527/1248=0.422276$. Omitting the final bath when it acts
away from the readout gives the stronger, post-design ideal corollary
$B_5(1)=1327/2272=0.584067$, by the same endpoint principle used in
Section~\ref{sec:gram}. These formulas reconstruct the archived benchmarks;
their model-specific proofs are not results asserted anew in this paper.
Exceeding the weaker target does not certify the exact-channel hypotheses.

\section{Reproducibility and certificate conventions}\label{app:reproducibility}
The complete reader package \texttt{reproducibility\_bundle.zip} contains the manuscript sources, compiled PDF, figures, bibliography, \texttt{build.py}, \texttt{README.md}, and the offline calculation and data files. The article source package \texttt{arxiv\_source.zip} contains the TeX sources, bibliography and publication figures, together with \texttt{anc/calculations\_and\_data.zip}. This nested archive contains the same calculation and data files without manuscript TeX. Each archive has an \texttt{ARCHIVE\_MANIFEST.json} recording SHA-256 file hashes; no archival identifier has been assigned to this draft.

The calculation package supplies \texttt{verify.py}, \texttt{make\_figures.py}, the five polynomial and size-certificate programs listed in the complete package's README, \texttt{anc/audit\_finite\_gram.py}, the frozen \texttt{anc/*.json} certificates, and \texttt{anc/exact\_gram.py}. The primary experiment is under \texttt{data/two\_qubit/}; the \texttt{live/} subdirectory contains the raw bit array \texttt{raw\_outcomes.npy} and the files \texttt{results.json}, \texttt{contexts.json}, \texttt{manifest.json}, and \texttt{submission.json}. All circuit and context files identified by the manifest are included. The \texttt{supporting\_reports/} directory supplies reports, job identifiers and the cancellation record for Appendix~\ref{c-app:earlier}, rather than those earlier studies' complete raw data.

Extract the complete reader package or the nested calculation archive into a new directory. With Python 3.10 or later, NumPy and Matplotlib installed, run
\begin{lstlisting}
python3 verify.py --full --figures
\end{lstlisting}
No credentials, device access, network connection or quantum SDK is needed. Without flags, the verifier checks the five polynomial and size certificates, saved exact moment identities, direct enumeration at length three and the two-qubit matrices. The \texttt{--full} option adds independent reconstruction of the Gram classes through suffix length seven, with rational enclosures of the eight-site moments. The \texttt{--figures} option recounts the primary shots, checks hashes and regenerates the figures and result table, creating \texttt{figures/} and \texttt{generated/} if absent. Logs and the list of executed checks are written to \texttt{review/verification/summary.json} and adjacent log files.

To rebuild the paper, run \texttt{python3 build.py} in the complete reader package with pdfLaTeX and BibTeX installed, or compile \texttt{main.tex} from the article source package. Loading the archived QPY circuits would require Qiskit, but no verification step loads them.

\subsection{Exact polynomial certificates}
The programs in \texttt{anc/} use integer and rational arithmetic.
To check nonnegativity on the simplex $q_1,q_2,q_3\ge0$, $\sum_iq_i=1$, a homogeneous
polynomial $P=\sum_{a+b+c=d}b_{abc}B_{abc}$ of degree $d$ is expanded in the simplex Bernstein basis
\begin{equation}
 B_{abc}(q)=\frac{d!}{a!b!c!}q_1^a q_2^b q_3^c,
 \qquad a+b+c=d.
\end{equation}
Each basis polynomial is nonnegative, and their sum is one. Consequently
the polynomial lies between its smallest and largest Bernstein
coefficients. For a permutation-symmetric polynomial it suffices to take $q_1\ge q_2\ge q_3$; ordered-cone certificates substitute
$(q_1,q_2,q_3)=(x+y+z,y+z,z)$, that is $x=q_1-q_2$, $y=q_2-q_3$, $z=q_3$; nonnegative
monomial coefficients certify nonnegativity for $x,y,z\ge0$, and strictly positive coefficients give the strict inequalities used in the appendices.
The appendices identify the substitutions, rational endpoints, and
coefficient counts for each use. Exponential and logarithmic estimates, such as the ceilings $\lceil700\log1710\rceil=5211$ and $\lceil700\log2736\rceil=5540$ of Proposition~\ref{prop:gramlocal}, use explicit
rational Taylor bounds with a bounded remainder and rational enclosures of the logarithm (\texttt{anc/certify\_geometry\_sizes.py}), rather than a
floating-point sign test.

\subsection{Reconstructing the finite-chain moments}
Here $\A_3$ is the thirteen-letter alphabet~\eqref{eq:alphabet}, $u_{\rm G}$ the normalized-Gram charge~\eqref{eq:gramcharge}, $B=a_2\cdots a_L$ the product of the last $L-1$ letters and $a$ the first letter, so $aB=G_L(w)$. For a suffix product $B$ of length $L-1$, put
\begin{equation}
 r_a(B)=u_{\rm G}(aB)-\tfrac13,\qquad
 f(B)=\frac1{13}\sum_{a\in\A_3}r_a(B),\qquad
 h_2(B)=\frac1{13}\sum_{a\in\A_3}r_a(B)^2.
\end{equation}
If $m(B)$ is the number of suffix words with product $B$, then
\begin{equation}
 w_L=13^{-(L-1)}\sum_Bm(B)f(B)^2,\qquad
 v_L=13^{-(L-1)}\sum_Bm(B)h_2(B),\qquad
 b_L^2=v_L-w_L.
\end{equation}
These formulas follow by conditioning the uniform word on its suffix:
\begin{equation*}
 13^{-L}\sum_wg(w)=13^{-(L-1)}\sum_s\frac1{13}\sum_ag(as).
\end{equation*}
The suffixes $s$ are grouped by $B=G_{L-1}(s)$. They are the $v$, $w$, $b^2$ of Section~\ref{sec:gram}: $f=m_1-\tfrac13$, $h_2$ averages to $m_2-\tfrac19$ because the exact mean of $u_{\rm G}$ is $1/3$, and $m(B)$ is $N_{L-1}(B)$ there; $b_L^2=v_L-w_L$ follows from orthogonality of the conditional expectation $F=PH_L^{\rm G}$.
Every summand is rational. Suffix products with the same Gram matrix $B^TB$ have the same multiset $\{u_{\rm G}(aB):a\in\A_3\}$ and the same successor Gram matrices up to relabelling the letters (Section~\ref{sec:gram}), so they can be grouped with multiplicities to reduce the
enumeration; the total multiplicity must remain $13^{L-1}$. This grouping is what \texttt{anc/audit\_finite\_gram.py} does.
The independent program \texttt{anc/audit\_finite\_gram.py} reconstructs
the Gram classes (223,093 at suffix length seven), checks $L=3$ by direct enumeration of all $13^3$ words, and verifies the saved eight-site inequalities $\lambda_8<1/8$, $v_8>1/20$, $w_8>7/160$ by rational enclosures. The complete
exact-generation program \texttt{anc/exact\_gram.py} (run with \texttt{--max-length 8}) is also included. This finite computation is a
proof certificate for the specified moments, not an extrapolation in $L$.

\subsection{Recounting the hardware estimator}
The four settings are those of Table~\ref{c-tab:results}, the score of a shot is $S=(-1)^{b_{\rm in}+b_{\rm out}}$ as in~\eqref{c-eq:estimator}, and $\widehat C$ estimates $C_F/w$. For each setting let $n_+$ and $n_-$ count scores $+1$ and $-1$.
Then $N_s=n_++n_-$ and
\begin{equation}
 \widehat C=\frac{n_+-n_-}{N_s},\qquad
 h=\sqrt{\frac{2\log(160)}{N_s}}.
\end{equation}
The second expression follows from Hoeffding's inequality, $2e^{-N_sh^2/2}$ per setting, and a union
bound over four settings: $8e^{-N_sh^2/2}=0.05$ exactly when $e^{N_sh^2/2}=160$; $\log$ is the natural logarithm and $h$ is the half-width of $[\widehat C-h,\widehat C+h]$.
The program \texttt{make\_figures.py} loads the archived $512\times128$ bit array and
setting labels, recomputes all scores and counts, checks the bit array against its post-acquisition SHA-256 hash and the unchanged circuit and context files against their pre-submission hashes, and regenerates the result table and figures. Public copies of the manifest and two execution scripts omit local execution labels; \texttt{PUBLICATION\_PROVENANCE.json} records their original and distributed hashes, and the program checks these copies against the latter. The original experiment hashes are retained, not replaced by publication hashes. Neither the no-reset echo
settings nor the earlier exploratory runs of Appendix~\ref{c-app:earlier} are used to divide or otherwise renormalize the primary estimator.

\section*{Statement on the use of artificial-intelligence tools}
The problem, the solution strategy and all calculations in this paper were worked out by the author on paper. Large language model systems, Claude (Anthropic) and GPT (OpenAI), were then used to verify every result by independent re-derivation, carried out by separately instantiated model instances; to prepare the manuscript text from the author's derivations; to search the literature and check the references; and to implement and analyse the hardware protocol under a pre-registered plan. No language model is an author, and the author takes full responsibility for every statement in this paper. This statement is made in accordance with the arXiv policy on reporting significant use of such tools.
{\small
\setlength{\bibsep}{3pt}
\bibliographystyle{unsrtnat}
\bibliography{references}
}
\end{document}